\documentclass[11pt]{article}
\date{\today}

\usepackage[dvipdfmx]{graphicx}
\usepackage{amsmath}
\usepackage{ascmac}
\usepackage{graphicx}
\usepackage{mathrsfs}
\usepackage{stmaryrd}
\usepackage{amsmath,amssymb}
\usepackage{latexsym}
\usepackage{cases}
\usepackage{url}
\usepackage{pifont}
\usepackage[dont-mess-around]{fnpct}
\usepackage{setspace}
\usepackage{arydshln}
\usepackage{here}
\usepackage{multirow}
\usepackage{natbib}
\usepackage[english]{babel}
\usepackage{threeparttable}
\usepackage[a4paper]{geometry}
\usepackage[linkcolor=blue,urlcolor=blue]{hyperref}
\usepackage{amsthm}
\hypersetup{
bookmarkstype=none,
colorlinks,
linkcolor=blue,
citecolor=black}
\usepackage{prettyref}

\newtheorem{theorem}{Theorem}[section]

\newtheorem{remark}{Remark}[section]

\newtheorem{lemma}[theorem]{Lemma}
\newtheorem{assumption}{Assumption}[section]

\newcommand{\MF}{\mathcal{F}}
\newcommand{\MJ}{\mathcal{J}}
\newcommand{\MG}{\mathcal{G}}

\newcommand{\MX}{\mathcal{X}}
\newcommand{\MY}{\mathcal{Y}}

\newcommand{\MB}{\mathcal{B}}

\newcommand{\Bt}{\mathbf{t}}
\newcommand{\By}{\mathbf{y}}
\newcommand{\Real}{\mathbb{R}}

\allowdisplaybreaks

  \makeatletter      
  \renewcommand{\@biblabel}[1]{}
  \makeatother

\def\spacingset#1{\renewcommand{\baselinestretch}%
{#1}\small\normalsize} \spacingset{1}

\numberwithin{equation}{section}
\usepackage{subfig}
\makeatother

\begin{document}
\setstretch{1.2}

\title{Partial Identification and Inference in Duration Models with Endogenous Censoring\thanks{I would like to thank the managing editor Elie Tamer, an anonymous associate editor, and an anonymous referee for their constructive suggestions and comments. This paper is based on a chapter of my dissertation at Kyoto University. I am grateful to Yoshihiko Nishiyama and Ryo Okui for their guidance and helpful suggestions. Parts of this paper were written while I was visiting Penn State as a visiting student. I would like to thank Keisuke Hirano for his hospitality and helpful advice. I would also like to thank Hidehiko Ichimura, Toru Kitagawa, Andrew Chesher, Francesca Molinari, Daniel Wilhelm, Jeff Rowley, and participants at various seminars and conferences for their comments and suggestions. I acknowledge financial support from JSPS KAKENHI Grant (number 18J00173) and ERC Grant (number
715940). Supplementary material is available \href{https://drive.google.com/file/d/16JKj67KfELo0hjuRnOKm55-yWpdavl_H/view?usp=sharing}{here}} 
}
\author{Shosei Sakaguchi\thanks{
Department of Economics, University College London, Gower Street, London WC1E 6BT, United Kingdom.  Email: s.sakaguchi@ucl.ac.uk.}
}
\date{\today}
\maketitle
\vspace{-0.6cm}
\begin{abstract}
This paper studies identification and inference in transformation models with endogenous censoring. Many kinds of duration models, such as the accelerated failure time model, proportional hazard model, and mixed proportional hazard model, can be viewed as transformation models. We allow the censoring of a duration outcome to be arbitrarily correlated with observed covariates and unobserved heterogeneity. We impose no parametric restrictions on either the transformation function or the distribution function of the unobserved heterogeneity. In this setting, we develop bounds on the regression parameters and the transformation function, which are characterized by conditional moment inequalities involving U-statistics. We provide inference methods for them by constructing an inference approach for conditional moment inequality models in which the sample analogs of moments are U-statistics. 
We apply the proposed inference methods to evaluate the effect of heart transplants on patients' survival time using data from the Stanford Heart Transplant Study.
\vspace{1.0em}
\begin{spacing}{1.1}
\end{spacing}
\end{abstract}

\newpage

\begin{spacing}{1.5}
\section{Introduction \label{sec:introduction}}
Duration models are widely used in various empirical studies in economics and biomedical sciences, where outcomes of interest are durations up to the occurrence of some event. Durations of interest in economics include unemployment duration, strike duration, insurance claim duration, and the duration until the purchase of a durable good.\footnote{\cite{van_den_Berg_2001} surveys many applications of duration models.} 
\par
In practice, duration data are often censored. For example, unemployment duration is likely to be censored due to some individuals dropping out of the survey, due to attrition say. Dealing with censoring has been a substantial challenge in duration analysis, and various methods have been proposed. A standard approach is to assume that censoring is independent of unobserved heterogeneity (conditional or unconditional on observed characteristics). Studies employing this approach include \cite{Cox_1972}, \cite{Powell_1984}, \cite{Ying_et_al_1995}, \cite{Yang_1999}, \cite{Honore_et_al_2002}, \cite{Hong_Tamer_2003}, and \cite{Khan_Tamer_2007}, among others. However, in many cases, justifying this independence assumption is difficult. For example, in unemployment duration analysis, unemployed individuals with low motivation to find a job may tend to drop out of the survey at an early stage. \cite{Szydlowski_2019} presents a number of examples where censoring is correlated with unobserved heterogeneity (i.e., censoring is endogenous). 
\par
In this paper, we study identification and inference in transformation models in the presence of endogenous censoring. The transformation model is expressed as
\begin{align}
Y^{\ast}=\Lambda(X^{\prime}\beta_{0}+U), \label{eq:transformation function}
\end{align}
where $\Lambda$ is a non-degenerate monotone function; $Y^{\ast}$ is a dependent variable, which represents a duration outcome in this paper; $X$ is a $k$-dimensional vector of observed covariates, where $k\geq 2$; $\beta_{0}$ denotes a $k$-vector of regression parameters; and $U$ is unobserved heterogeneity that is independent of $X$. Many kinds of duration models, such as the accelerated failure time model, proportional hazard model, and mixed proportional hazard (MPH) model, can be viewed as transformation models.\footnote{Aside from duration models, a class of transformation models contains other important kinds of models, for example, the linear index model and Box-Cox transformation model.} In this paper, we consider a nonparametric transformation model in which neither the transformation function nor the distribution function of the unobserved heterogeneity is parametrically specified. One important model represented by the nonparametric transformation model is the nonparametric MPH model in which neither a baseline hazard function nor the distribution function of the unobserved heterogeneity is parametrically specified. 
\par
Allowing for endogenous censoring, we develop bounds on the regression parameters $\beta_{0}$ and transformation function $\Lambda$ in model (\ref{eq:transformation function}). 
To the best of our knowledge, this is the first work to derive bounds on $\beta_{0}$ and $\Lambda$ in the nonparametric transformation model with endogenous censoring.
The construction of the bounds is built on the rank properties of the nonparametric transformation model studied by \cite{Han_1987} and \cite{Chen_2002}. \cite{Han_1987} shows that if there is no censoring, at least one element of $X$ has full-support on the real line, and $X$ is full-rank, the regression parameters are point-identified up to scale by looking at the rank correlation between the outcomes and covariates. In the presence of endogenous censoring, we develop bounds on $\beta_0$ by supposing that, in the rank property studied by \cite{Han_1987}, each censored outcome takes an infinitely large value or is equal to the censoring time. This reflects the fact that concerning each censored outcome, all we know is that it may take any value larger than censoring time. When $\Lambda$ is strictly increasing, a lower bound on $\Lambda$ is also attained by incorporating endogenous censoring into the rank property studied by \cite{Chen_2002}. Once the bounds on the regression parameters and the transformation function are obtained, we can also derive bounds on the distribution function of the unobserved heterogeneity $U$.
\par
The bounds on $\beta_0$ and $\Lambda$ are characterized by conditional moment inequalities whose sample moments are U-statistics. Based on these conditional moment inequalities, we construct inference methods for these parameters by extending the inference approach of \cite{Andrews_Shi_2013} for conditional moment inequality models to the case of U-statistics. The proposed inference approach can be applied not only to this work but also to other works involving conditional moment inequalities and U-statistics. In this sense, this paper also contributes to the literature on inference for conditional moment inequality models.\footnote{Various inference methods for conditional moment inequality models have been proposed, for example, by \cite{Andrews_Shi_2013, Andrews_Shi_2014, Andrews_Shi_2017}, \cite{Chernozhukov_et_al_2013}, \cite{Armstrong_2014, Armstrong_2015}, \cite{Menzel_2014}, \cite{Chernozhukov_et_al_2019}, and so on. But none of them can be applied to sample moment functions of U-statistics.}
\par
The bounded sets of the parameters proposed in this paper are not necessarily sharp identified sets. On the other hand, using concepts from random set theory (e.g., \cite{Beresteanu_et_al_2011, Beresteanu_et_al_2012}), we also characterize the sharp identified set of the regression parameters. However, constructing a feasible inference method based on it is difficult, whereas the proposed sets are tractable to construct feasible inference methods. In the paper, we also discuss conditions under which the proposed set of $\beta_0$ approaches the sharp identified set.


This paper is mostly related to works that study endogenous censoring. \cite{Khan_Tamer_2009}, \cite{Khan_et_al_2011, Khan_et_al_2016}, \cite{Li_Oka_2015}, and \cite{Fan_Liu_2018} study identification and estimation of parameters in quantile regression models with endogenous censoring. For cross-sectional linear quantile regression models, \cite{Khan_Tamer_2009} provide a point identification result for the linear coefficients under a certain support condition, while \cite{Khan_et_al_2011} provide a partial identification result without this support condition. Under censoring characterized by a certain copula, \cite{Fan_Liu_2018} partially identify the linear coefficients of the same model. \cite{Li_Oka_2015} and \cite{Khan_et_al_2016} consider panel quantile regression models with endogenous censoring and provide partial identification results. In contrast to these works, the identification result in this paper does not rely on quantile modeling, copula characterization of censoring, or panel data. Aside from quantile models, \cite{Szydlowski_2019} considers the parametric MPH model and proposes a sharp identified set and inference method for its parameters. While \cite{Szydlowski_2019} considers the parametric MPH model, we consider the nonparametric one, which is robust to misspecification of the hazard function or the distribution function of unobserved heterogeneity. 

\par
For competing risks models, \cite{Honore_Lleras_Muney_2006} partially identify the parameters in the accelerated failure time model, and \cite{Kim_2018} derives computationally tractable bounds on distributions of latent durations by exploiting the discreteness of observed durations. In this paper, we allow for continuous observed durations and do not specify competing risks. In a sample selection model, \cite{Honore_Hu_2020} obtain the sharp identified set of linear regression coefficients by imposing a particular structure on the sample selection.

\par
The remainder of this paper is structured as follows. Section \ref{sec:model and identification} describes the setup and assumptions and then provides the main results to develop bounds on the regression parameters and the transformation function. We also characterize the sharp identified set of the regression parameters and compare it to our proposed superset. Section \ref{sec:inference} provides an inference method for the regression parameters and derives its asymptotic properties. A joint inference method for the regression parameters and the transformation function is presented in Appendix \ref{app:joint_inference}. Section \ref{sec:simulation study} presents numerical examples and Monte Carlo simulation results. The numerical examples show how the bounds on the regression parameters and transformation function vary depending on the degree of censoring and the support of covariates. The Monte Carlo simulation results show the finite sample properties of the proposed inference method for the regression parameters. Section \ref{sec:empirical illustration} presents an empirical illustration, where we apply our proposed inference methods to evaluate the effect of heart transplants on patients' survival duration using data from the Stanford Heart Transplant Study. 
We conclude this paper with some remarks in Section \ref{sec:conclusion}. All proofs are presented in Appendices \ref{app:proof_1} and \ref{app:proof_2}. 
Some additional numerical examples and Monte Carlo simulation results are presented in the supplementary material to this paper.

\section{Model and Identification \label{sec:model and identification}}

We first describe the setting of the paper and provide conditions to develop bounds on the regression parameters in Section \ref{sec:model}. Subsequently, in Section \ref{sec:partial identification}, we present the main result to construct bounds on the regression parameters. In Section \ref{sec:sharp identified set}, we characterize the sharp identified set of the regression parameters using concepts from random set theory, and compare this with our proposed superset.
Section \ref{sec:transformation function} derives bounds on the transformation function and the distribution of the unobserved heterogeneity. 

\subsection{Model \label{sec:model}}

We consider the transformation model in the form of (\ref{eq:transformation function}). In the model, we do not specify the transformation function $\Lambda$ or the distribution function of the unobserved heterogeneity, which we denote by $F_{U}$. Because of this, we impose location and scale normalizations. For the location normalization, we suppose that the constant term is equal to zero (i.e., $X$ does not contain a constant term). For the scale normalization, we suppose that the absolute value of the first component of $\beta_{0}$ is equal to one (i.e., $\left|\beta_{0,1}\right|=1$), where $\beta_{0,j}$ denotes the $j$-th component of $\beta_{0}$. Later, in Section \ref{sec:transformation function}, we impose an additional location normalization to fix $\Lambda$. Let $B\equiv \{-1,1\}\times \Real^{k-1}$ denote the normalized regression parameter space. Our first main focus is on the identification of and inference on the normalized regression parameters $\beta_{0}$ in $B$.
\par
The transformation model contains many kinds of duration models as its special cases: the accelerated failure time model, Cox's proportional hazard model, and the MPH model.\footnote{If $\Lambda(Y^{\ast})=\exp( Y^{\ast})$, the transformation model corresponds to the accelerated failure time model; if $\Lambda(Y^{\ast})=\exp( \Delta(Y^{\ast}))$, where $\Delta(\cdot)$ is the integrated baseline hazard function and $U$ has the CDF $F(u)=1-\exp(-e^{u})$, the transformation model corresponds to Cox's proportional hazard model; if  $\Lambda(Y^{\ast})=\exp(\Delta(Y^{\ast}))$ and $U=\epsilon + \nu$ where $\nu$ is unobserved heterogeneity and $\epsilon$ has the CDF $F(\epsilon)=1-\exp(-e^{\epsilon})$, the transformation model corresponds to the MPH model. For more details, see \citeauthor{Horowitz_2009} (\citeyear{Horowitz_2009}, Ch. 6).} In particular, the nonparametric MPH model is an important duration model represented by a nonparametric transformation model. The MPH model extends Cox's proportional hazard model by incorporating individual unobserved heterogeneity. Since introduced in \cite{Lancaster_1979}, the MPH model has been widely used in various empirical studies in economics. In the nonparametric MPH model, the normalized regression parameters $\beta_{0}$ can be interpreted as the logs of the scale-normalized hazard ratios (see, e.g., \cite{Lancaster_1990}).
\par
When data are subject to censoring, the duration outcome $Y^{\ast}$ cannot always be observed. Instead, for unit $i=1,\ldots,n$, we observe $W_{i}=\left(Y_{0i},D_{i},X_{i}\right)$ such that $Y_{0i}=\mbox{min}\left\{ Y_{i}^{\ast},C_{i}\right\} $ and $D_{i}=I\left[Y_{i}^{\ast}\leq C_{i}\right]$, where $C_{i}$ is a random censoring variable and $I\left[\cdot\right]$ denotes the indicator function. $D_{i}$ is a censoring indicator that takes the value zero if $Y_{i}^{\ast}$ is censored and the value one if $Y_{i}^{\ast}$ is observed. Note that we consider right censoring in the paper, but all the results presented below are easily extendable to left and interval censoring. Using $D_{i}$, $Y_{0i}$ can be expressed as $Y_{0i}=D_{i}Y_{i}^{\ast}+(1-D_{i})C_{i}$. Let $P$ denote the distribution function of a vector of random variables $\left(X,U,C\right)$ and $\MX \subseteq \Real^{k}$ denote the support of $X$.
\par
Throughout this paper, we suppose that the following assumptions hold.
\bigskip
\begin{assumption}\label{asm:monotonicity}
$\Lambda$ is a non-degenerate monotonically increasing function.
\end{assumption}

\begin{assumption}\label{asm:iid}
The vectors $\left(Y_{i}^{\ast},C_{i},X_{i}\right)$, $i=1,\ldots,n$, are independent and identically distributed (i.i.d) as $(Y^\ast, X,C)$ where $Y^\ast$ is distributed according to the latent transformation model (\ref{eq:transformation function}) and the support of $X$ is $\MX$.
\end{assumption}

\begin{assumption} \label{asm:independence}
$U$ is independent of $X$.
\end{assumption}
\bigskip


Note that Assumption \ref{asm:independence} does not restrict the relationship between $U$ and $C$, allowing for endogenous censoring.\footnote{\cite{Chiappori_et_al_2015} study identification and estimation of the nonparametric transformation model when some covariates are endogenous and there is no censoring.} 

\subsection{Identification of the Regression Parameters \label{sec:partial identification}}

This section constructs bounds on the regression parameters $\beta_{0}$.
The construction is based on a rank property of the latent outcome $Y^{*}$ and the regression part $X^\prime \beta_{0}$ in the transformation model (\ref{eq:transformation function}).  
For explanatory purposes, we first introduce the point identification result of \cite{Han_1987} in the absence of censoring. 
\par
\cite{Han_1987} supposes that there is no censoring (i.e., $Y^{\ast}$ is always observed). In this case, under Assumptions \ref{asm:monotonicity}--\ref{asm:independence}, a full-support condition on an element of $X$, a full-rank condition on $X$, and a continuous distribution of $U$, he shows that $\beta_{0}$ uniquely satisfies the following rank property,
\begin{equation}
x_{i}^{\prime}\beta_{0}\geq x_{j}^{\prime}\beta_{0}\Rightarrow P(Y_{i}^{\ast}\geq Y_{j}^{\ast}\mid x_{i},x_{j})\geq \frac{1}{2} \label{eq:rank inequality}
\end{equation}
for all $(x_{i},x_{j})\in \mathcal{X}^{2}$, where $P(\cdot \mid x_{i},x_{j})$ denotes the conditional probability given $(X_i,X_j)=(x_i,x_j)$.\footnote{\cite{Han_1987} actually considers a slightly different rank property. Theorem H.1 in the supplementary material to this paper shows that $\beta_0$ uniquely satisfies (\ref{eq:rank inequality}) for all $(x_{i},x_{j})\in \mathcal{X}^{2}$ under the same assumptions as in his theorem.}
This rank property means that, for any given pair of $(x_{i},x_{j})$, the probability that $Y_{i}^{\ast}$ is larger than or equal to $Y_{j}^{\ast}$ is greater than or equal to 1/2 if and only if $x_{i}^{\prime}\beta_{0}$ is larger than or equal to $x_{j}^{\prime}\beta_{0}$. Then $\beta_{0}$ is the unique value in $B$ that satisfies this rank relation for any pair of $(x_{i},x_{j}) \in \MX^{2}$. In other words, for any $\beta\neq\beta_{0}$, there exists at least one pair $(x_{i},x_{j}) \in \MX^{2}$ that violates the rank relation (\ref{eq:rank inequality}). 
\par
In this paper, we suppose that censoring exists and it may be endogenous. Hence, we cannot always observe $Y_{i}^{\ast}$ and do not have any information about the censoring mechanism. The censoring variable $C_{i}$ may be arbitrarily correlated with the observed covariates $X_{i}$ and unobserved heterogeneity $U_{i}$. 
\par
In this situation, we can still construct bounds on the regression parameters $\beta_{0}$. Let $Y_{1i}\equiv D_{i}Y_{i}^{\ast}+(1-D_{i})(+\infty)$, which is an outcome variable that takes an arbitrary large value when the primary outcome is censored. Recall that $Y_{0i}=D_{i}Y_{i}+(1-D_{i})C_{i}$. Then because $P(Y_{1i}\geq Y_{0j}\mid x_{i},x_{j})\geq P(Y_{i}^{\ast}\geq Y_{j}^{\ast}\mid x_{i},x_{j})$
holds for all $(x_{i},x_{j}) \in \MX^{2}$, the following rank property holds
from (\ref{eq:rank inequality}):
\begin{equation}
x_{i}^{\prime}\beta_{0}\geq x_{j}^{\prime}\beta_{0}\Rightarrow P(Y_{1i}\geq Y_{0j}\mid x_{i},x_{j})\geq\frac{1}{2} \label{eq:rank inequality with censoring}
\end{equation}
for all $(x_{i},x_{j}) \in \mathcal{X}^{2}$. Therefore, defining
\begin{equation*}
B_{I}\equiv \{\beta\in B\mid \ \mbox{$x_{i}^{\prime}\beta\geq x_{j}^{\prime}\beta\Rightarrow P\left(Y_{1i}\geq Y_{0j}\mid x_{i},x_{j}\right)\geq\frac{1}{2}$ for all \mbox{\ensuremath{(x_{i},x_{j})\in\mathcal{X}^{2}}}}\}, 
\end{equation*}
$\beta_{0}$ is contained in $B_{I}$. This set is derived from a worst-case analysis where we suppose that censored outcomes may take extreme values, $C$ or $+\infty$, for any given value of $x$. This reflects the fact that concerning each censored outcome, all we know is that it may take any value at least larger than its censored time. 

The following assumption ensures that $B_I$ is a proper subset of $B(=\{-1,1\}\times \Real^{k-1})$.
\bigskip

\begin{assumption}\label{asm:support condition}
Let $\widetilde{\MX}^{2} \equiv \{(x_i,x_j)\in \MX^2 \mid P(Y_{1i}\geq Y_{0j}\mid x_{i},x_{j})<\frac{1}{2}\}$. Then $P((X_i,X_j) \in \widetilde{\MX}^{2})>0$.
\end{assumption}

\bigskip

The following theorem summarizes the main result in this section.\bigskip

\begin{theorem} \label{thm:partial identification} 
Under Assumptions \ref{asm:monotonicity}--\ref{asm:support condition}, $\beta_{0}\in B_{I} \subset B$ a.s. 
\end{theorem}
\bigskip
A proof of this theorem is given in Appendix \ref{app:proof_1}. We make several remarks about this theorem.
First, in this theorem, we do not impose a full-support condition on an element of $X$, in contrast to many works in the semiparametric literature, as we no longer focus on point identification.\footnote{\cite{Magnac_Maurin_2008}, \cite{Blevins_2011}, and \cite{Komarova_2013} discuss the difficulties of justifying the full-support condition in a number of cases, and provide partial identification results for different semiparametric models in the absence of the full-support condition.}
A full-rank condition on $X$ is necessary (but not sufficient) for $B_I$ to be bounded.
Second, as we will see in the following subsection, $B_{I}$ is not a sharp identified set. However, this set is easy to compute and to construct a feasible inference method, as we will see in Section \ref{sec:inference}. The following subsection shows how the sharp identified set can be characterized, and its computational difficulty.




\subsection{Characterization of the Sharp Identified Set \label{sec:sharp identified set}}

In this section, we illustrate a way to characterize the sharp identified set using concepts from random set theory. Subsequently, we compare the sharp identified set with $B_{I}$. This comparison clarifies why $B_{I}$ is not a sharp set and in which situations $B_I$ approaches the sharp set. For the definitions and notations for random set theory used in this section, see, for example, \cite{Molchanov_2005} or \citeauthor{Beresteanu_et_al_2012} (\citeyear{Beresteanu_et_al_2012}, Appendix A). Throughout this section, for any variable $A$, we denote by $\tilde{A}$ an independent copy of $A$.
\par
Using concepts from random set theory, we can characterize the incomplete information for the latent outcome variable $Y^{\ast}$. For any random variable $A$, let $A_{x}$ denote a random variable that has the conditional distribution of $A$ given $X=x$. Then, for a given $x\in\mathcal{X}$, what we observe for the latent outcome variable in the presence of endogenous censoring can be expressed as the random set $\mathcal{Y}_{x}$ defined as
\begin{eqnarray*}
\mathcal{Y}_{x} & = & \begin{cases}
\begin{array}{c}
\{Y_{x}^{\ast}\}\\
(C_{x},+\infty)
\end{array} & \begin{array}{c}
\mbox{if}\ D_{x}=1\\
\mbox{otherwise}
\end{array}\end{cases},
\end{eqnarray*}
where $Y_{x}^{\ast}$ and $C_{x}$ are, respectively, a latent outcome and a censoring variable given $X=x$;  $D_{x}=I\left[Y_{x}^{\ast}\leq C_{x}\right]$ is a censoring indicator given $X=x$. 
Hence, all the information for the latent outcome variable can be expressed by stating that $Y_{x}^{\ast}\in\mbox{Sel}(\mathcal{Y}_{x})$.\footnote{For any random set ${\cal Y}$, a random variable $Y$ is called a measurable selection of ${\cal Y}$ if $Y\in{\cal Y}$ a.s., and $\mbox{Sel}\left({\cal Y}\right)$ is defined to be the set of all measurable selections of ${\cal Y}$. See, for example, \citeauthor{Molchanov_2005} (\citeyear{Molchanov_2005}, Ch. 1) or \citeauthor{Beresteanu_et_al_2012} (\citeyear{Beresteanu_et_al_2012}, Appendix A).} Let $B_{0}$ denote the sharp identified set of $\beta_{0}$. Throughout this section, we suppose that the full-support condition on one element of $X$, the full-rank condition on $X$, and the continuity of $F_U$ hold to ensure the sharp identification result.\footnote{More specifically, we suppose that the following three conditions hold: 
(i) $\MX$ is not contained in any proper linear subspace of $\Real^k$;
(ii) for almost every $x_{-1} = (x_2\ \ldots, x_k)$, the distribution of $X_{(1)}$ conditional on $(X_{(2)},\ldots,X_{(k)})=x_{-1}$ has an everywhere positive density, where $X_{(m)}$ denotes the $m$-th element of $X$; 
(iii) $F_U$ is a continuous distribution function.
\label{asm:full-support and rank conditions}
}\footnote{These conditions might not be needed to derive the sharp identified set; however, to our knowledge, there is no work that derives the sharp identified set for the nonparametric transformation model in the absence of these conditions.}
\par
Combining the above random set representation with \citeauthor{Han_1987}'s (\citeyear{Han_1987}) point identification result that $\beta_0$ uniquely satisfies the the rank property (\ref{eq:rank inequality}) for all $(x_i,x_j) \in \MX^2$, the sharp identified set $B_{0}$ is characterized as the set of $\beta$ such that there exists a family of pairs of selections $(Y_{x_{i}},\tilde{Y}_{x_{j}})\in\mbox{Sel}(\mathcal{Y}_{x_{i}})\times \mbox{Sel}(\widetilde{\MY}_{x_j})$ over $(x_{i},x_{j})\in {\cal{X}}^{2}$ that satisfy the following:
\begin{equation}
x_{i}^{\prime}\beta\geq x_{j}^{\prime}\beta\Leftrightarrow P(Y_{x_{i}}\geq \tilde{Y}_{x_{j}})\geq\frac{1}{2} \label{eq:rank inequality_random set}
\end{equation}
for all $(x_{i},x_{j})\in{\cal X}^{2}$, where $\widetilde{\MY}_{x}$ is the random set of $\tilde{Y}_{x}$. Therefore, $B_{0}$ is characterized as
\begin{equation}
B_{0}=\left\{ \beta\in B\mid \exists \left\{(Y_{x_i},\tilde{Y}_{x_j})\in \mbox{Sel}(\MY_{x_{i}})\times
\mbox{Sel}(\widetilde{\MY}_{x_{j}})\right\}_{(x_{i},x_{j})\in {\cal{X}}^{2}},\forall(x_{i},x_{j})\in{\cal X}^{2}, \mbox{(\ref{eq:rank inequality_random set}) holds}\right\}. \label{eq:sharp identified set}
\end{equation}
\par
We next look at how the proposed set $B_{I}$ can be characterized by the random set. Given $x \in \MX$, by definition, $Y_{1x}$ and $Y_{0x}$ satisfy (i) $Y_{1x},Y_{0x}\in\mbox{Sel}(\mathcal{Y}_{x})$ a.s. and (ii) $Y_{0x} \leq Y_{x}\leq Y_{1x}$ a.s. for all $Y_{x}\in \mbox{Sel}(\mathcal{Y}_{x})$. 
Thus for any given pair $(x_{i},x_{j}) \in \MX^2$, the parameter set 
\begin{equation*}
\{\beta\in B\mid \mbox{$x_{i}^{\prime}\beta\geq x_{j}^{\prime}\beta\Rightarrow P\left(Y_{1i}\geq Y_{0j}\mid x_{i},x_{j}\right)\geq\frac{1}{2}$}\}
\end{equation*}
 is equivalent to 
\begin{equation*}
\left\{ \beta\in B\mid \exists\left(Y_{x_{i}},\tilde{Y}_{x_{j}}\right)\in \mbox{Sel}(\mathcal{Y}_{x_{i}})\times\mbox{Sel}(\widetilde{\MY}_{x_{j}}),\mbox{(\ref{eq:rank inequality_random set}) holds}\right\}, 
\end{equation*}
which is the set of $\beta$ such that for the fixed $(x_{i},x_{j})$, there exists a pair of selections $\left(Y_{x_{i}},\tilde{Y}_{x_{j}}\right)\in\mbox{Sel}(\mathcal{Y}_{x_{i}})\times\mbox{Sel}(\widetilde{\MY}_{x_j})$ that satisfies the inequality (\ref{eq:rank inequality_random set}). Therefore, $B_{I}$ is characterized as the set of $\beta$ such that for any pair $(x_{i},x_{j})\in {\cal{X}}^{2}$, there exists a pair of selections $\left(Y_{x_{i}},\tilde{Y}_{x_{j}}\right)\in \mbox{Sel}(\mathcal{Y}_{x_{i}})\times\mbox{Sel}(\widetilde{\MY}_{x_{j}})$ that satisfy the inequality (\ref{eq:rank inequality_random set}). Formally, $B_{I}$ is characterized as
\begin{equation}
B_{I}=\left\{ \beta\in B\mid \forall(x_{i},x_{j})\in{\cal X}^{2}, \exists\left(Y_{x_{i}},\tilde{Y}_{x_{j}}\right)\in \mbox{Sel}(\mathcal{Y}_{x_{i}})\times\mbox{Sel}(\widetilde{\MY}_{x_{j}}),\mbox{(\ref{eq:rank inequality_random set}) holds}\right\}. \label{eq:proposed identified set}
\end{equation}
\par
The difference between (\ref{eq:sharp identified set}) and (\ref{eq:proposed identified set}) (i.e., the different orders of ``$\forall(x_{i},x_{j})\in{\cal X}^{2}$'' and ``$\exists\left(Y_{x_{i}},\tilde{Y}_{x_{j}}\right)\in \mbox{Sel}(\mathcal{Y}_{x_{i}})\times\mbox{Sel}(\widetilde{\MY}_{x_{j}})$'' in (\ref{eq:sharp identified set}) and (\ref{eq:proposed identified set})) shows that $B_{0}$ is contained in $B_{I}$, but $B_0$ does not necessary contain $B_I$. Hence $B_{I}$ is not necessarily a sharp set. Some intuition for the non-sharpness of $B_{I}$ is as follows. Fix a triple $\left(x_{i},x_{j},x_{k}\right) \in \MX^3$ and $\beta \in B$ such that $x_{i}^{\prime}\beta \leq x_{j}^{\prime}\beta \leq x_{k}^{\prime}\beta$.   
When we check  whether $\beta \in B_{I}$ through the rank property (\ref{eq:rank inequality}), in comparing $x_{j}$ and $x_{k}$, we suppose that the latent outcome variable $Y_{x_{j}}^{\ast}$ takes its smallest value, $C_{x_{j}}$; whereas, in comparing $x_{j}$ and $x_{i}$, we suppose that $Y_{x_{j}}^{\ast}$ takes its largest value, $+\infty$. However, when we check whether $\beta \in B_0$ through the characterization of (\ref{eq:sharp identified set}),
we compare fixed selections $Y_{x} \in \mbox{Sel}(\mathcal{Y}_{x})$ over all $x\in \cal{X}$; that is, $Y_{x_{j}}$ does not change when comparing different $\tilde{Y}_{x_{i}}$ over $x_{i}\in \cal{X}$. This difference explains why $B_{I}$ is larger than $B_{0}$. 

\bigskip

\begin{remark}
Although we could characterize the sharp set $B_{0}$ as (\ref{eq:sharp identified set}), it is hard to compute. When examining whether a certain $\beta$ is contained in $B_{0}$, one would have to search for the existence of measurable selections $Y_{x}\in\mbox{Sel}(\mathcal{Y}_{x})$ for all $x\in \cal{X}$ that satisfy the rank inequality (\ref{eq:rank inequality_random set}) for all pairs $(x_{i},x_{j})\in\mathcal{X}^{^{2}}$. In contrast, $B_{I}$ is much easier to compute. For this reason, we focus on $B_{I}$ in this paper rather than the sharp set.\footnote{
\cite{Beresteanu_et_al_2011, Beresteanu_et_al_2012} suggest using the support function and Aumann expectation to make it easy to compute the sharp identified set. However, when we follow this approach, we still have to search for the measurable selections $Y_{x}\in\mbox{Sel}(\mathcal{Y}_{x})$, for all $x\in \cal{X}$, to satisfy a certain equality. Thus this approach does not greatly ease the computation of the sharp identified set in our setting.\bigskip
}
\end{remark}

\begin{remark}
 There are some situations when $B_{I}$ is close to $B_{0}$. The first case is when censoring does not occur frequently. If censoring is unlikely to occur given any $x\in \mathcal{X}$, measurable selections of $\mathcal{Y}_{x}$ correspond to a single measurable selection $Y_{x}^{\ast}$ with high probability. The second case is when $Y_{x}^{\ast}$ is not censored at small values for each $x \in \MX$; that is, $C_{x}$ takes a large value when $Y_{x}^{\ast}$ is censored. In this case, for each $x \in \MX$, the random interval $(C_{x},+\infty)$ in the definition of $\MY_x$ is narrow, and hence the difference between (\ref{eq:sharp identified set}) and (\ref{eq:proposed identified set}) does not make much difference between $B_{I}$ and $B_{0}$. In the empirical example of the heart transplantation study in Section \ref{sec:empirical illustration}, this case corresponds to the case when each patient is unlikely to drop out of the study at an early stage.
\end{remark} 
\bigskip

\subsection{Identification of the Transformation Function \label{sec:transformation function}}

We next develop a lower bound on the transformation function in the presence of endogenous censoring. Knowing about the transformation function enables us to infer the type of duration model. We here suppose that $\Lambda$ is strictly monotonic. 
\bigskip

\begin{assumption}\label{asm:strict monotonicity}
$\Lambda$ is a strictly increasing function.
\end{assumption}
\bigskip

Let $T\equiv \Lambda^{-1}$. We call $T$ the transformation function when no confusion arises.  
As an additional location normalization to fix $T$, we suppose $T(\tilde{y})=0$ for some specific outcome value $\tilde{y}<\infty$.
We now construct a bound on the normalized $T(y)$ at a particular value of $y\in\mathbb{R}$.
\par
Construction of the bound is built on the rank property studied by \cite{Chen_2002}. 
In the case of no censoring, provided that Assumptions \ref{asm:iid}--\ref{asm:independence} and $\ref{asm:strict monotonicity}$ hold and that the true regression parameters $\beta_{0}$ are given,\footnote{Under the supposed conditions with the full-support and full-rank conditions on $X$ and the continuity of $F_U$, $\beta_{0}$ can be point identified by, for example, applying \citeauthor{Han_1987}'s \citeyearpar{Han_1987} maximum rank correlation approach.} \cite{Chen_2002} shows that $T\left(y\right)$ satisfies the following rank property:
\begin{equation}
x_{i}^{\prime}\beta_{0}-x_{j}^{\prime}\beta_{0}\geq T(y) \Rightarrow P\left(Y_{i}^{\ast}\geq y\mid x_{i}\right)\geq P\left(Y_{j}^{\ast}\geq\tilde{y}\mid x_{i}\right) \label{eq:chen's rank property}
\end{equation}
for all $(x_{i},x_{j})\in{\cal X}^{2}$, where we recall that $\tilde{y}$ is such that $T\left(\tilde{y}\right)=0$ for the location normalization. 
Moreover, supposing further that the full-support and full-rank conditions on $X$ and the continuity of $F_U$ hold, \cite{Chen_2002} shows that $T\left(y\right)$ can be point identified as the minimum value of $t$ that satisfies the inequality (\ref{eq:chen's rank property}) with $t$ in place of $T(y)$ for all $(x_i,x_j)\in \MX^2$. \cite{Chen_2002} also provides an inference method based on this identification result. 
\par
In the presence of endogenous censoring, we can construct a lower bound on $T\left(y\right)$ using a similar idea to that presented in Section \ref{sec:partial identification}. If $\beta_{0}$ were given, because $P(Y_{1i}\geq y\mid x_{i})\geq P(Y_{i}^{\ast}\geq y\mid x_{i})$
and $P(Y_{j}^{\ast}\geq y\mid x_{j})\geq P(Y_{0j}\geq y\mid x_{j})$ hold for all $(x_{i},x_{j}) \in \MX^2$, it follows from (\ref{eq:chen's rank property}) that $T\left(y\right)$ is contained in the following set:
\begin{equation}
\left\{ t\in\mathbb{R}:x_{i}^{\prime}\beta_{0}-x_{j}^{\prime}\beta_{0}\geq t\Rightarrow P\left(Y_{1i}\geq y\mid x_{i}\right)\geq P\left(Y_{0j}\geq\tilde{y}\mid x_{i}\right)\ \mbox{for all}\ (x_{i},x_{j})\in{\cal X}^{2}\right\} . \label{eq:identified set_transformation function}
\end{equation}
In the presence of endogenous censoring, we cannot point identify $\beta_{0}$; instead, we have $B_{I}$ which contains $\beta_{0}$. Thus, letting
\begin{equation}
T_{I,\beta}\left(y\right) \equiv \left\{ t\in\mathbb{R}:x_{i}^{\prime}\beta-x_{j}^{\prime}\beta\geq t\Rightarrow P\left(Y_{1i}\geq y\mid x_{i}\right)\geq P\left(Y_{0j}\geq\tilde{y}\mid x_{i}\right)\ \mbox{for all}\ (x_{i},x_{j})\in{\cal X}^{2}\right\}
\label{eq:definition_T_I_beta}
\end{equation}
and $T_{B_I}\left(y\right) \equiv \left\{ T_{I,\beta}\left(y\right)\mid\beta\in B_I\right\}$,
we have that $T\left(y\right)\in T_{B_{I}}\left(y\right)$ a.s. 
Note that since any $t>T(y)$ satisfies all the inequality conditions in (\ref{eq:definition_T_I_beta}), $T_{I,\beta}(y)$, as well as $T_{B_I}$, does not have a finite upper bound. Hence we can only obtain a lower bound on $T(y)$.
Note also that, due to a similar reason as that discussed in Section \ref{sec:sharp identified set}, $T_{I,\beta_{0}}\left(y\right)$ is not a sharp identified set of $T\left(y\right)$ even if $\beta_{0}$ is known. Hence, $\left\{ T_{I,\beta}\left(y\right)\mid\beta\in B_0\right\}$ is not a sharp identified set even if we had $B_{0}$.

The following theorem formalizes the main result in this subsection. 
\bigskip
\begin{theorem} \label{thm:transformation function}
Under Assumptions \ref{asm:iid}--\ref{asm:independence} and \ref{asm:strict monotonicity}, $T\left(y\right)\in T_{B_{I}}\left(y\right)$ holds a.s. for any $y\in\mathbb{R}$.
\end{theorem}
\bigskip

The proof is provided in Appendix \ref{app:proof_1}. 
For any fixed $y \in \Real$, $T_{B_I}(y)$ has a finite lower bound if there exists at least one pair of $(x_i,x_j) \in \MX^2$ such that $P\left(Y_{1i}\geq y\mid x_{i}\right)< P\left(Y_{0j}\geq\tilde{y}\mid x_{i}\right)$ holds
and that $x_i^{\prime}\beta - x_j^{\prime}\beta>-\infty$ holds for all $\beta \in B_I$.
The lower bound approaches the true value $T(y)$ as the likelihood of censoring diminishes.



\bigskip

\begin{remark}
Once we obtain the bounds on $\beta_{0}$ and $T$, we can derive a bound on the error distribution $F_{U}$. Let $u \in \Real$ be fixed. Because $U=T(Y^{\ast})-X^{\prime}\beta_{0}$ holds and $U$ is independent of $X$, $F_U(u)$ is equivalent to $P(T(Y^{\ast})-X^{\prime}\beta_{0}\leq u \mid X=x)$ for all $x \in \MX$. 
 Due to the censoring, we cannot always observe $Y^{\ast}$, but we instead have 
 \begin{align*}
     P(T(Y_{1})-X^{\prime}\beta_{0}\leq u \mid X=x) \leq F_U(u) \leq P(T(Y_{0})-X^{\prime}\beta_{0}\leq u \mid X=x),
 \end{align*}
because $T$ is a non-decreasing function. In the presence of endogenous censoring, we cannot point identify $\beta_0$ and $T$, but we instead have the sets $B_I$ for $\beta_{0}$ and $T_{I,\beta}(\cdot)$ for $T(\cdot)$ with any $\beta \in B_I$. Therefore, $F_U(u)$ is contained in the following set:
 \begin{align*}
     \{ t \in [0,1]:\ & P(T(Y_{1})-X^{\prime}\beta <u \mid X=x) \leq t \leq P(T(Y_{0})-X^{\prime}\beta <u \mid X=x)\\
     &\mbox{ holds for some } (\beta,T) \in B_I \times T_{I,\beta} \mbox{ and all }x \in \MX\},
 \end{align*}
 where $T_{I,\beta}$ denotes a set of all functions $T$ that satisfy $T(y) \in T_{I,\beta}(y)$ for all $y$.
\end{remark}
\bigskip

\section{Inference \label{sec:inference}}

This section provides a statistical inference approach for the regression parameters in model (\ref{eq:transformation function}) based on the result presented in Section \ref{sec:partial identification}. We suggest a method to construct a confidence set that covers the true parameter value $\beta_{0}$ with a probability greater than or equal to $1-\alpha$ for $\alpha\in\left(0,1\right)$. Because $B_{I}$ is characterized by conditional moment inequalities involving U-statistics, we construct the inference method by extending the inference approach for conditional moment inequality models proposed by \cite{Andrews_Shi_2013} (hereafter AS) to the U-statistics case. The approach transforms conditional moment inequalities into an infinite number of unconditional ones, without information loss, to construct a test statistic. A confidence set is then constructed by inverting the test statistic and using critical values obtained via moment selection.
The inference method proposed below is applicable to either continuous or discrete covariates.
The inference method is for U-statistics of order two, but the approach is extendable to U-statistics of greater order with obvious modifications. A joint inference procedure for $\beta_{0}$ and $T(\cdot)$ is presented in Appendix \ref{app:joint_inference}.

\subsection{Test Statistic and Critical Value \label{sec:test statistics}}

We first construct a test statistic and then describe the inference procedure. 
Let 
\begin{align*}
m\left(W_{i},W_{j},\beta\right) \equiv -\frac{1}{2}+I[Y_{1i}\geq Y_{0j}]\cdot I[X_{i}^{\prime}\beta\geq X_{j}^{\prime}\beta]+I[Y_{1j}>Y_{0i}]\cdot I[X_{j}^{\prime}\beta>X_{i}^{\prime}\beta],
\end{align*}
where $W_{i} \equiv (Y_{1i},Y_{0i},X_{i})$ for $i=1,\ldots,n$.
Then $B_{I}$ is a set of parameter values $\beta$ that satisfy the following conditional moment inequalities:
\begin{equation}
E_{P}\left[m\left(W_{i},W_{j},\beta\right)\mid x_{i},x_{j}\right]\geq0\ \mbox{for a.e.}\ (x_{i},x_{j})\in\mathcal{X}^{2}, \label{eq:conditional moment inequality}
\end{equation}
where $E_P[\cdot \mid x_i,x_j]$ denotes the conditional expectation under the distribution $P$ given $(X_i,X_j)=(x_i,x_j)$.
To transform all of the conditional moment inequalities (\ref{eq:conditional moment inequality}) into unconditional ones without loss of information, we adopt AS's instrumental functions approach. 
We here suppose that the first $p$ ($p \leq k$) elements of $X$ are continuous variables or have infinite support and that the remaining elements of $X$ have finite support.\footnote{We can also allow all covariates to have finite or infinite support with modifications.}
Let $\MX_{1}$ and $\MX_{2}$ denote the supports of $(X_{i,1},\ldots,X_{i,p})$ and $(X_{i,p+1},\ldots,X_{i,k})$, respectively, where $X_{i,j}$ denotes the $j$-th element of $X_i$. Without loss of generality, we also suppose that $(X_{i,1},\ldots,X_{i,p})$ is transformed via a one-to-one mapping so that each of its elements lies in $[0,1]$ (i.e., $\MX_{1} \subseteq \left[0,1\right]^{p}$).
\footnote{Let $X_{i}^{(1)}\equiv(X_{i,1},\ldots,X_{i,p})$. Following AS, the vector of transformed covariates in $X_{i}^{(1)}$ may be $\tilde{X}_{i}^{(1)}\equiv\Phi\left(\hat{\Sigma}_{X^{(1)},n}^{-1/2}\left(X_{i}^{(1)}-\bar{X}_{n}^{(1)}\right)\right)$
where $\bar{X}_{n}^{(1)}\equiv n^{-1}\sum_{i=1}^{n}X_{i}^{(1)}$, $\hat{\Sigma}_{X^{(1)},n}\equiv n^{-1}\sum_{i=1}^{n}\left(X_{i}^{(1)}-\bar{X}_{n}^{(1)}\right)\left(X_{i}^{(1)}-\bar{X}_{n}^{(1)}\right)^{\prime}$, and $\Phi\left(x^{(1)}\right)\equiv\left(\Phi\left(x_{1}\right),\ldots,\Phi\left(x_{p}\right)\right)^{\prime}$, where $\Phi\left(\cdot\right)$ denotes the standard normal cumulative distribution function and $x^{(1)}=\left(x_{1},\ldots,x_{p}\right)^{\prime}$.}

The set of instrumental functions that we consider is of the following form:
\begin{align*}
\mathcal{G}  \equiv  \left\{ g(x_{i},x_{j})=I\left[(x_{i},x_{j}) \in J \right]\ \mbox{for}\ J\in\mathcal{J} \right\} ,
\end{align*}
where
\begin{align*}
\mathcal{J}  \equiv & \left\{ J_{(a,b),(\tilde{a},\tilde{b}),r}=\left(\vartimes_{u=1}^{p}\left(\frac{a_{u}-1}{2r},\frac{a_u}{2r}\right] \times \{b\} \right) \times \left(\vartimes_{u=1}^{p}\left(\frac{\tilde{a}_{u}-1}{2r},\frac{\tilde{a}_{u}}{2r}\right] \times \{\tilde{b}\}\right):\right.\\
   &\ \left. a=(a_{1},\ldots,a_{p}),\ \tilde{a}=(\tilde{a}_{1},\ldots,\tilde{a}_{p}),\ (a_u,\tilde{a}_u) \in\left\{ 1,2,\ldots,2r\right\}^2 \right. \\
   & \left.\mbox{ for } u=1,\ldots,p \mbox{ and } r=1,2,\ldots, \mbox{ and } (b,\tilde{b})\in \MX_2^2 \right\} .
\end{align*}
This set of instrumental functions transforms the conditional moment inequalities (\ref{eq:conditional moment inequality}) into infinitely many unconditional ones without loss of information. Accordingly, under Assumptions \ref{asm:iid}--\ref{asm:independence}, $B_{I}$ is equivalent to
\begin{equation*}
\left\{ \beta\in B:\ E_{P}\left[m(W_{i},W_{j},\beta,g)\right]\geq0\ \mbox{for all}\ g\in\mathcal{G}\right\} ,
\end{equation*}
where $m(W_{i},W_{j},\beta,g)\equiv m\left(W_{i},W_{j},\beta\right)\cdot g(x_{i},x_{j})$ for $g\in{\cal G}$. We formalize this result as Lemma \ref{lem:lemma_2} in Appendix \ref{app:proof_2}. Other kinds of instrumental functions introduced in AS are applicable with modifications. 
\par
We define the sample moment function and sample variance function of $m(W_{i},W_{j},\beta,g)$, respectively, as
\begin{align*}
\bar{m}_{n}\left(\beta,g\right)  \equiv  \frac{1}{n(n-1)}\sum_{i\neq j}m(W_{i},W_{j},\beta,g)
\end{align*}
and
\begin{align*}
\hat{\sigma}_{n}^{2}\left(\beta,g\right) \equiv & \left\{ \frac{1}{n\left(n-1\right)\left(n-2\right)}\sum_{i\neq j\neq k}m(W_{i},W_{j},\beta,g)m\left(W_{i},W_{k},\beta,g\right)\right.\\
   & \left.-\left(\frac{1}{n(n-1)}\sum_{i\neq j}m(W_{i},W_{j},\beta,g)\right)^{2}\right\} .
\end{align*}
Note that $\bar{m}_{n}\left(\beta,g\right)$ and $\hat{\sigma}_{n}^{2}\left(\beta,g\right)$ are U-statistics of orders two and three, respectively. Because $\bar{m}_{n}\left(\beta,g\right)$ is a non-degenerate U-statistic of order two, the asymptotic variance
of $\sqrt{n}\bar{m}_{n}\left(\beta,g\right)$ is $\mbox{Var}_{P}\left(E_{P}\left[m(W_{i},W_{j},\beta,g)\mid W_{i}\right]\right)$,
which is equivalent to\footnote{For the variance of U-statistics, see, for example, \citeauthor{van_der_Vaart_1998} (\citeyear{van_der_Vaart_1998}, Ch. 12).}
\begin{equation*}
E_{P}\left[m(W_{i},W_{j},\beta,g)m\left(W_{i},W_{k},\beta,g\right)\right]-\left(E_{P}\left[m(W_{i},W_{j},\beta,g)\right]\right)^{2}.
\end{equation*}
 Thus, $\hat{\sigma}_{n}^{2}\left(\beta,g\right)$ is a consistent estimator of the asymptotic variance of $\sqrt{n}\bar{m}_{n}\left(\beta,g\right)$. However, in practice, $\hat{\sigma}_{n}^{2}\left(\beta,g\right)$ could be zero for some $g \in \mathcal{G}$; as such we use the modification proposed by AS for $\hat{\sigma}_{n}^{2}\left(\beta,g\right)$. The modified version of $\hat{\sigma}_{n}^{2}\left(\beta,g\right)$ is 
\begin{align*}
\bar{\sigma}_{n}^{2}\left(\beta,g\right) \equiv \hat{\sigma}_{n}^{2}\left(\beta,g\right)+\epsilon\hat{\sigma}_{n}^{2},
\end{align*}
where $\hat{\sigma}_{n}^{2}=\hat{\sigma}_{n}^{2}\left(\beta,1\right)$, which is a consistent estimator of
\begin{align*}
\sigma_{P}^{2}\left(\beta\right)\equiv E_{P}\left[m\left(W_{i},W_{j},\beta\right)m\left(W_{i},W_{k},\beta\right)\right]-\left(E_{P}\left[m\left(W_{i},W_{j},\beta\right)\right]\right)^{2},
\end{align*}
and $\epsilon$ is a regularization parameter that takes some fixed positive value. In the simulation studies in Section \ref{sec:simulation study} and empirical application in Section \ref{sec:empirical illustration}, we use $\epsilon=0.0001$ and $\epsilon =0.001$.
\par
Then, letting $g_{(a,b),(\tilde{a},\tilde{b}),r}(x_{i},x_{j})\equiv1\left[(x_{i},x_{j})\in J_{(a,b),(\tilde{a},\tilde{b}),r} \right]$,
the test statistic at $\beta$ takes the form 
\begin{align*}
T_{n}(\beta)  \equiv  \sum_{r=1}^{\infty}\left(r^{2}+100\right)^{-1}
\sum_{\left(a,\tilde{a}\right)\in \left\{ 1,\cdots,2r\right\}^{2p}}
\sum_{(b,\tilde{b})\in \MX_{2}^{2}}
\frac{1}{\left((2r)^{p} \cdot |\MX_{2}|\right)^2} \left[\frac{n^{\frac{1}{2}}\bar{m}_{n}\left(\beta,g_{(a,b),(\tilde{a},\tilde{b}),r}\right)}{\bar{\sigma}_{n}\left(\beta,g_{(a,b),(\tilde{a},\tilde{b}),r}\right)}\right]_{-}^{2},
\end{align*}
where $\left[x\right]_{-}=-x$ if $x<0$ and $\left[x\right]_{-}=0$ if $x\geq0$. 
$\left((2r)^{p} \cdot (|\MX|)^{k-p}\right)^{2}$ corresponds to the number of instrumental functions $g_{(a,b),(\tilde{a},\tilde{b}),r}$ for a fixed $r$.
This test statistic is a version of AS's test statistic that is extended to U-statistics of order two. Here the inner summation is taken over two set of indices, $(a,b)$ and $(\tilde{a},\tilde{b})$. In the implementation, we instead use an approximate test statistic at $\beta$:
\begin{align*}
T_{n,R}(\beta)  \equiv  \sum_{r=1}^{R}\left(r^{2}+100\right)^{-1}
\sum_{\left(a,\tilde{a}\right)\in \left\{ 1,\cdots,2r\right\}^{2p}}
\sum_{(b,\tilde{b})\in \MX_{2}^{2}}
\frac{1}{\left((2r)^{p} \cdot |\MX_{2}|\right)^2} \left[\frac{n^{\frac{1}{2}}\bar{m}_{n}\left(\beta,g_{(a,b),(\tilde{a},\tilde{b}),r}\right)}{\bar{\sigma}_{n}\left(\beta,g_{(a,b),(\tilde{a},\tilde{b}),r}\right)}\right]_{-}^{2},
\end{align*}
where $R$ is some truncation integer chosen by the researcher. 
\par
To compute a critical value for $T_{n,R}(\beta)$, we propose using an asymptotic approximation version of the critical value. This is a simulated quantile of 
\begin{align*}
T_{n,R}^{Asy}(\beta)  \equiv & \sum_{r=1}^{R}\left(r^{2}+100\right)^{-1}
\sum_{(a,\tilde{a})\in \{ 1,\cdots,2r\}^{2p}}
\sum_{(b,\tilde{b})\in \MX_{2}^{2}}
\left\{ \frac{1}{\left((2r)^{p} \cdot |\MX_{2}|\right)^2} \right. \\ 
& \left. \times
\left[\frac{v_{n}\left(\beta,g_{(a,b),(\tilde{a},\tilde{b}),r}\right)+\varphi_{n}\left(\beta,g_{(a,b),(\tilde{a},\tilde{b}),r}\right)}{\bar{\sigma}_{n}\left(\beta,g_{(a,b),(\tilde{a},\tilde{b}),r}\right)}\right]_{-}^{2} \right\},
\end{align*}
where $\left(v_{n}\left(\beta,g\right)\right)_{g\in{\cal G}}$ is a zero mean Gaussian process with a covariance kernel evaluated by
\begin{align*}
\hat{h}_{2}\left(\beta,g,g^{\ast}\right)  \equiv & \left\{ \frac{1}{n\left(n-1\right)\left(n-2\right)}\sum_{i\neq j\neq k}m(W_{i},W_{j},\beta,g)m(W_{i},W_{k},\beta,g^{\ast})\right.\\
   & \left.-\left[\frac{1}{n\left(n-1\right)}\sum_{i\neq j}m(W_{i},W_{j},\beta,g)\right]\cdot\left[\frac{1}{n\left(n-1\right)}\sum_{i\neq j}m\left(W_{i},W_{j},\beta,g^{\ast}\right)\right]\right\} .
\end{align*}
In the expression of $T_{n,R}^{Asy}(\beta)$, $\left(v_{n}\left(\beta,g_{(a,b),(\tilde{a},\tilde{b}),r}\right)\right)_{(a,b),(\tilde{a},\tilde{b}),r}$
approximates the asymptotic distribution of 
\begin{equation*}
\left(n^{\frac{1}{2}}\left[\bar{m}_{n}\left(\beta,g_{(a,b),(\tilde{a},\tilde{b}),r}\right)-E_{P}\left[m\left(W_{i},W_{j},\beta,g_{(a,b),(\tilde{a},\tilde{b}),r}\right)\right]\right]\right)_{(a,b),(\tilde{a},\tilde{b}),r}.
\end{equation*}
$\varphi_{n}\left(\beta,g_{(a,b),(\tilde{a},\tilde{b}),r}\right)$ is a generalized moment selection (GMS) function to select binding moment restrictions and is given by
\begin{equation*}
\varphi_{n}\left(\beta,g_{(a,b),(\tilde{a},\tilde{b}),r}\right) \equiv \hat{\sigma}_{n}^{2}B_{n}I\left[\kappa_{n}^{-1}n^{\frac{1}{2}}\bar{m}_{n}\left(\beta,g_{(a,b),(\tilde{a},\tilde{b}),r}\right)/\bar{\sigma}_{n}\left(\beta,g_{(a,b),(\tilde{a},\tilde{b}),r}\right)>1\right],
\end{equation*}
where $B_{n}$ and $\kappa_{n}$ are two tuning parameters that should satisfy $B_{n}\rightarrow\infty$, $\kappa_{n}\rightarrow\infty$, and $\kappa_{n}/n^{1/2}\rightarrow 0$ as $n\rightarrow\infty$ a.s. In Sections \ref{sec:simulation study} and \ref{sec:empirical illustration}, we use $B_{n}=\left(0.8\ln\left(n\right)/\ln\ln\left(n\right)\right)^{\frac{1}{2}}$ and $\kappa_{n}=\left(\left(1-\hat{p}_{1-D}^{1/3}\right)^{2/5}\times0.6\ln(n)\right)^{\frac{1}{2}}$, where $\hat{p}_{1-D}\equiv n^{-1}\sum_{i}^{n}\left(1-D_{i}\right)$ is the sample censoring rate.\footnote{The values for $\kappa_{n}$ and $B_{n}$ are different from the values recommend
by AS, which are $\kappa_{n}=\left(0.3\ln(n)\right)^{\frac{1}{2}}$ and $B_{n}=\left(0.4\ln\left(n\right)/\ln\ln\left(n\right)\right)^{\frac{1}{2}}$, respectively.} This value of $\kappa_{n}$ decreases with the sample censoring rate. The following assumption summarizes the requirements for the tuning parameters in the GMS function.
\bigskip
\begin{assumption} \label{asm:tuning parameters}
The tuning parameters $\left(B_{n},\kappa_{n}\right)$ satisfy $B_{n}\rightarrow\infty$, $\kappa_{n}\rightarrow\infty$, and $\kappa_{n}/n^{1/2}\rightarrow0$ as $n\rightarrow\infty$ a.s.
\end{assumption}
\bigskip
For a significance level of $\alpha<1/2$, the critical value is set to be the $1-\alpha+\eta$ simulated quantile of $T_{n,R}^{Asy}(\beta)$, where $\eta$ is an arbitrarily small positive value (e.g., $10^{-6}$ following AS). Letting $\hat{c}_{n,\eta,1-\alpha}(\beta)$ be the $1-\alpha+\eta$ quantile of $T_{n,R}^{Asy}(\beta)$, a nominal level $1-\alpha$ confidence set is computed as 
\begin{equation*}
\widehat{CS}_{n,\eta,1-\alpha}\equiv \left\{ \beta\in B:T_{n,R}(\beta)\leq\hat{c}_{n,\eta,1-\alpha}(\beta)\right\} .
\end{equation*}

\bigskip

\begin{remark}
When all the covariates have finite support, the above procedure is applicable using the test statistic
\begin{align*}
T_{n}(\beta)  =  
\sum_{(b,\tilde{b})\in \MX^{2}}
\frac{1}{|\MX|^2}
\left[\frac{n^{\frac{1}{2}}\bar{m}_{n}\left(\beta,g_{b,\tilde{b}}\right)}{\bar{\sigma}_{n}\left(\beta,g_{b,\tilde{b}}\right)}\right]_{-}^{2}
\end{align*}
and a simulated quantile of 
\begin{align*}
T_{n}^{Asy}(\beta)  =  
\sum_{(b,\tilde{b})\in \MX^{2}}
\frac{1}{|\MX|^2}
\left[\frac{v_{n}\left(\beta,g_{b,\tilde{b}}\right)+\varphi_{n}\left(\beta,g_{b,\tilde{b}}\right)}{\bar{\sigma}_{n}\left(\beta,g_{b,\tilde{b}}\right)}\right]_{-}^{2},
\end{align*}
where $g_{b,\tilde{b}}(x_i,x_j)=I[(x_i,x_j)=(b,\tilde{b})]$. In this case, inference is free from the tuning parameter $R$.

Regardless of whether some covariates have finite support or not, the above procedure with $\mathcal{J}$ replaced by
\begin{align*}
\mathcal{J}^{\prime}  \equiv & \left\{ J_{a,\tilde{a},r}=\left(\vartimes_{u=1}^{k}\left(\frac{a_{u}-1}{2r},\frac{a_u}{2r}\right] \right) \times \left(\vartimes_{u=1}^{k}\left(\frac{\tilde{a}_{u}-1}{2r},\frac{\tilde{a}_{u}}{2r}\right] \right):\right.\\
   &\ \left. a=(a_{1},\ldots,a_{k}),\ \tilde{a}=(\tilde{a}_{1},\ldots,\tilde{a}_{k}),\ (a_u,\tilde{a}_u) \in\left\{ 1,2,\ldots,2r\right\}^2 \right. \\
   & \left.\mbox{ for } u=1,\ldots,k \mbox{ and } r=1,2,\ldots \right\} .
\end{align*}
is also applicable when $\MX \subseteq [0,1]^{k}$. 
When $\MX_2$ does not contain too many points, the original inference procedure using $\mathcal{J}$ is computationally more efficient. In contrast, when $\MX_2$ contains many points, the inference procedure using $\mathcal{J}^{\prime}$ with a moderate number of $R$ incurs a lower computational cost.
\end{remark}

\bigskip

Note again that the inference approach presented above is for U-statistics of order two, but can be extended to U-statistics of a greater order with some modification. Namely, the class of instrumental functions, the moment function, and the sample variance function need to be modified for higher order; the inner double summations in the test statistics need to be replaced by more summations.

\subsection{Asymptotic Size and Power Properties \label{sec:asymptotic property}}

This subsection provides the uniform asymptotic size and power properties of the inference method. Let ${\cal Q}$ be the collection of all pairs of the regression parameters in $B$ and distributions, $(\beta,P)$, such that the conditional moment inequality condition (\ref{eq:conditional moment inequality}) holds, $\{W_i:i \geq 1\}$ are i.i.d. under $P$, and the support of $X$ is contained in $\MX(=\MX_{1}\times\MX_{2})$. Define 
\begin{align}
h_{2,P}(\beta,g,g^{\ast}) \equiv & E_{P}\left[E_{P}\left[m(W_{i},W_{j},\beta,g)m(W_{i},W_{k},\beta,g^{\ast})\mid W_{i}\right]\right] \notag \\
  & -E_{P}\left[m(W_{i},W_{j},\beta,g)\right]E_{P}\left[m\left(W_{i},W_{j},\beta,g^{\ast}\right)\right], \label{eq:covariance kernel}
\end{align}
which is the covariance kernel between $m(W_{i},W_{j},\beta,g )$ and $m(W_{i},W_{j},\beta,g^{\ast} )$ under the distribution $P$. Let ${\cal H}_{2}\equiv \{h_{2,P}(\beta,\cdot,\cdot):(\beta,P)\in \mathcal{Q}\}$ be the collection of all possible covariance kernel functions on ${\cal G}\times{\cal G}$. For any $\beta \in B$ and the true distribution $P$, let 
\begin{align*}
    \MX_{P}^{2}(\beta) \equiv \{(x_i,x_j) \in \MX^{2}: x_{i}^{\prime}\beta\geq x_{j}^{\prime}\beta \Rightarrow P(Y_{1i}\geq Y_{0j}) < 1/2\}.
\end{align*}
The following theorem gives the uniform size and power properties of the proposed inference method.
\bigskip

\begin{theorem} \label{thm:asymptotic property}
Suppose that $R=\infty$ and $\alpha<1/2$. \\
(a) For every compact subset ${\cal H}_{2,cpt}$ of ${\cal H}_{2}$, the confidence set $\widehat{CS}_{n,\eta,1-\alpha}$ satisfies 
\begin{equation*}
\underset{\eta\rightarrow0}{\lim}\liminf_{n\rightarrow\infty}\inf_{\left\{ \left(\beta,P\right)\in {\cal Q}:h_{2,P}(\beta,\cdot,\cdot)\in{\cal H}_{2,cpt}\right\} }P\left(\beta\in\widehat{CS}_{n,\eta,1-\alpha}\right)=1-\alpha.
\end{equation*}
(b) Suppose further that Assumptions \ref{asm:monotonicity}--\ref{asm:independence} and \ref{asm:tuning parameters} hold for the true distribution $P$, and let $\tilde{\beta}\in B$ be a vector of regression parameters such that $P((X_i,X_j) \in \MX_{P}^{2}(\tilde{\beta}))>0$. Then $\lim_{n\rightarrow\infty}P(\tilde{\beta}\in\widehat{CS}_{n,\eta,1-\alpha})=0$.
\end{theorem}
\bigskip

A proof of this theorem is provided in Appendix \ref{app:proof_2}. Theorem \ref{thm:asymptotic property} (a) states that the proposed confidence set is asymptotically conservative, which corresponds to Theorem 2(b) of AS. The uniformity in the statement enables the asymptotic result to provide a good finite sample approximation, which is well discussed in AS. Theorem \ref{thm:asymptotic property} (b) states that the test is consistent against a fixed alternative. 


\section{Simulation Studies \label{sec:simulation study}}

This section presents numerical examples and Monte Carlo simulation results. The numerical examples show how $B_{I}$ and $T_{B_I}$ vary with the degree of censoring and the support of covariates. The Monte Carlo simulations show the finite sample performance of the proposed inference method for the regression parameters and demonstrate how it varies with the choice of the tuning parameters. Some additional numerical example and Monte Carlo simulation studies are presented in the supplementary material to this paper.

\subsection{Numerical Examples \label{sec:numerical example}}

This subsection illustrates how $B_{I}$ and $T_{B_{I}}$ vary depending on the degree of censoring and the support of covariates. We consider three MPH models (Models 1--3) with endogenous censoring that have the following form:
\begin{eqnarray*}
\log Y^{\ast} & = & \beta_{1}X_{1}+\beta_{2}X_{2}+\log U+\log V,\\
\log C & = & \alpha_{0}+\left(\gamma_{0}+\gamma_{1}X_{1}+\gamma_{2}X_{2}\right)\times\log U+\log W.
\end{eqnarray*}
In all the models, we set $\left(\beta_{1},\beta_{2}\right)=\left(0.5,1.5\right)$ and $\left(\gamma_{0},\gamma_{1},\gamma_{2}\right)=\left(-0.5,0.5,-1\right)$.
In Models 1--3, we set $\alpha_{0}=+\infty$, $3$, and $1.6$, respectively. In Model 1, there is no censoring; in Models 2 and 3, there is censoring correlated with the covariates and unobserved heterogeneity. The outcome is more likely to be censored in Model 3 than in Model 2. In all the models, $U$, $V$, and $W$ have independent unit exponential distributions, and $X_{2}$ takes values in $\left\{ 0,1\right\} $. Regarding $X_{1}$, we consider three cases; $X_{1}$ takes values in (i) $\left\{ -2.5,-2.0,\ldots,2.5\right\} $, (ii) $\left\{ -5,-2.5,\ldots,5\right\} $, or (iii) $\left\{ -5,-4.8,\ldots,5\right\} $. 
In these data generating processes (DGPs), the censoring is endogenous because $Y^\ast$ is correlated with $C$ even conditional on $X_1$ and $X_2$, which occurs due to the presence of $U$ in both the equations for $\log Y^\ast$ and $\log C$. 
\par
In the computation of $B_I$ and $T_{B_I}$, we set $|\beta_1|=1$ for the scale normalization and $T(\tilde{y})=0$ for the location normalization, where we set $\tilde{y}=0.77$, which is the median of $Y^\ast$ when $X_{1}$ is independently distributed as $N(0,2)$ and $X_{2}$ is independently and uniformly distributed over $\{0,1\}$. This joint distribution will be used in the Monte Carlo simulation in the following subsection.
After the scale and location normalizations, the true normalized values of $\beta_2$ and $T(y)$ are $3$ and $T_{0}(y):=2(\log y - \log 0.77)$, respectively. 
\par
Given the DGPs described above, $B_{I}$ is the set of regression parameter values that satisfy the conditional moment inequality (\ref{eq:conditional moment inequality}) for all pair of values of $\left(X_{1},X_{2}\right)$, and $T_{B_{I}}(y)$ for a fixed $y \in \Real$ is the set of values $t$ that satisfy the inequality in (\ref{eq:definition_T_I_beta}) for some $\beta \in B_{I}$ and all pair of values of $(X_{1},X_{2})$.
We numerically obtain $B_{I}$ and $T_{B_I}(y)$ by simulating the distributions of $\log Y^\ast$ and $\log C$ for each pair of values of $\left(X_{1},X_{2}\right)$ using 20,000 random draws from each DGP. 
\par
Table \ref{tab:numerical example} presents the numerical results for $B_{I}$. Each cell in the table shows the interval of $\beta_2$ obtained by the projection of the computed $B_I$
for each model and each support of $X_{1}$. As expected, the interval shrinks as the censoring rate decreases or the support becomes wider or finer. In each model, expanding the support from (i) to (ii) does not shrink the lower or upper bounds on $\beta_2$. 

\bigskip
\begin{table}[h] 
\caption{Computed $B_{I}$ for Models 1--3 and Supports (i)--(iii)}
\vspace*{-0.5cm}
\begin{center}
\begin{tabular}{cccc} 
\hline 
Model / Support of $X_{1}$ & Support (i) & Support (ii) & Support (iii)\tabularnewline
\hline 
Model 1 & $[2.51,3.49]$ & $[2.51,3.49]$ & $[2.81,3.19]$\tabularnewline
Model 2 & $[2.00,4.00]$ & $[2.00,3.49]$ & $[2.21,3.49]$\tabularnewline
Model 3 & $[1.50,5.00]$ & $[1.50,3.99]$ & $[1.80,3.79]$\tabularnewline
\hline 
\end{tabular}
\begin{tablenotes}
\ \footnotesize Note: Each cell shows the interval of $\beta_2$ obtained by the projection of the computed $B_I$. 
\end{tablenotes}
\label{tab:numerical example}
\end{center}
\end{table}

Figure \ref{fig:id set for T} illustrates the computed $T_{B_{I}}$ for each model and each support of $X_{1}$, along with the true normalized transformation function $T_0$. As mentioned in Section \ref{sec:transformation function}, $T_{B_{I}}$ is not bounded from above. The result shows that as the censoring rate decreases or the support of $X_1$ becomes wider or finer, the lower bound on $T_{B_{I}}$ becomes closer to the true transformation function. The lower bound also becomes smoother as the support of $X_1$ becomes finer.
\par
The supplementary material to this paper gives two additional numerical examples for $B_{I}$. In the first example, we compute $B_{I}$ in models with three covariates, where the results in Table \ref{tab:numerical example} are extended to two-dimensional figures.
In the second example, using slightly different models and supposing that the transformation model is parametrically specified (i.e., the functional forms of $\Lambda$ and $F_{U}$ are known), we compute the sharp identified set of $\beta_0$, based on the result in \cite{Szydlowski_2019}, and compare it with $B_{I}$.

\bigskip
\begin{figure}[h!]
\caption{Computed $T_{B_I}$ for Models 1-3 and Supports (i)-(iii)}
\begin{tabular}{ccc}
\begin{minipage}[t]{0.33\hsize}
\begin{center}
{\small (a-1) Model 1 \& Support (i)}
\includegraphics[scale=0.27]{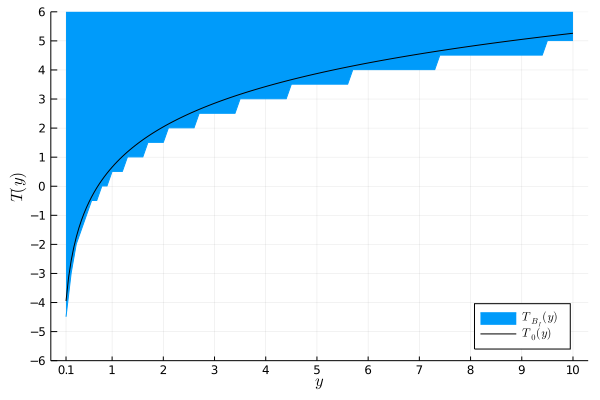}
\end{center}
\end{minipage}
\begin{minipage}[t]{0.33\hsize}
\begin{center}
{\small (a-2) Model 1 \& Support (ii)}
\includegraphics[scale=0.27]{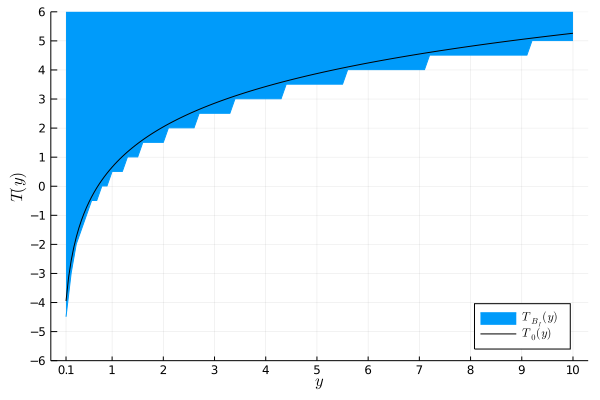}
\end{center}
\end{minipage}
\begin{minipage}[t]{0.33\hsize}
\begin{center}
{\small (a-3) Model 1 \& Support (iii)}
\includegraphics[scale=0.27]{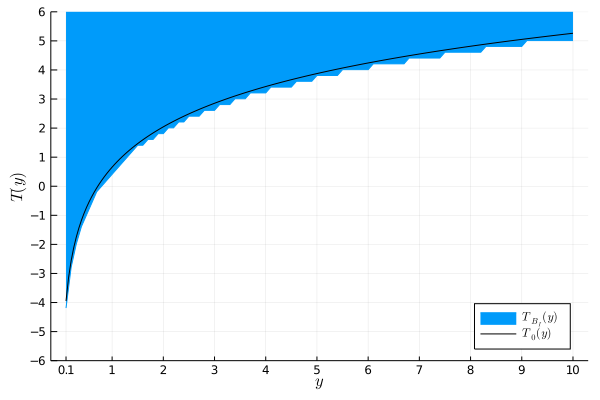}
\end{center}
\end{minipage}
\bigskip
\\
\begin{minipage}[t]{0.33\hsize}
\begin{center}
{\small (b-1) Model 2 \& Support (i)}
\includegraphics[scale=0.27]{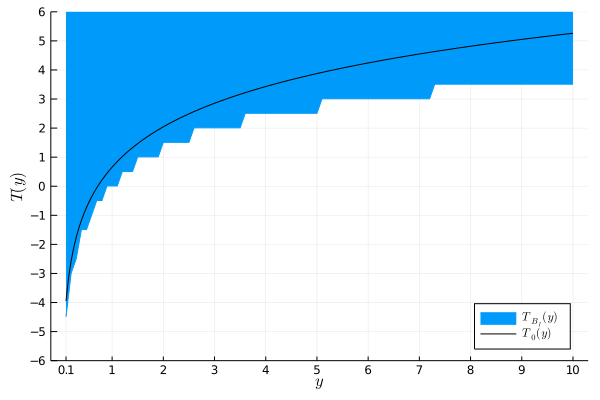}
\end{center}
\end{minipage}
\begin{minipage}[t]{0.33\hsize}
\begin{center}
{\small (b-2) Model 2 \& Support (ii)}
\includegraphics[scale=0.27]{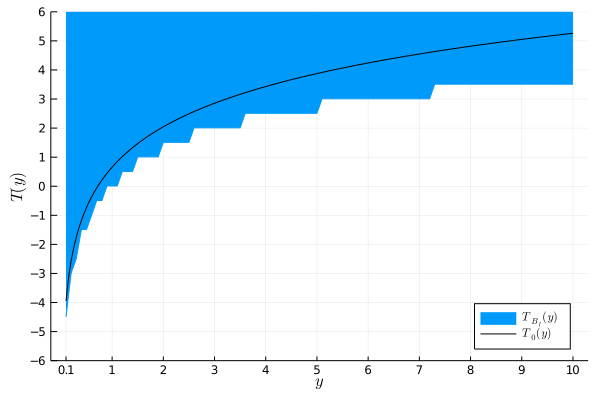}
\end{center}
\end{minipage}
\begin{minipage}[t]{0.33\hsize}
\begin{center}
{\small (b-3) Model 3 \& Support (iii)}
\includegraphics[scale=0.27]{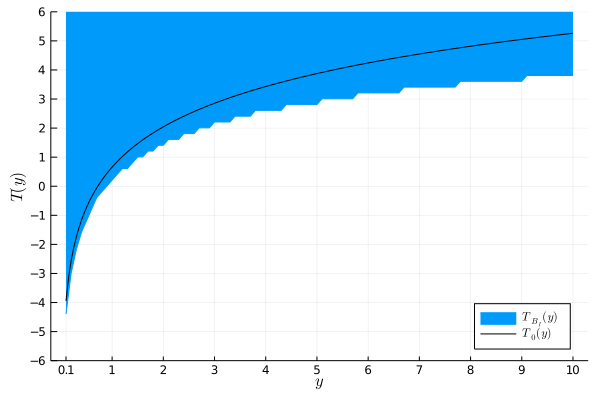}
\end{center}
\end{minipage}

\bigskip
\\
\begin{minipage}[t]{0.33\hsize}
\begin{center}
{\small (c-1) Model 3 \& Support (i)}
\includegraphics[scale=0.27]{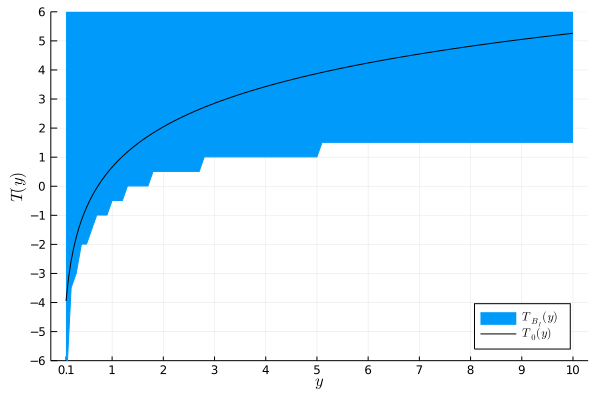}
\end{center}
\end{minipage}
\begin{minipage}[t]{0.33\hsize}
\begin{center}
{\small (c-2) Model 3 \& Support (ii)}
\includegraphics[scale=0.27]{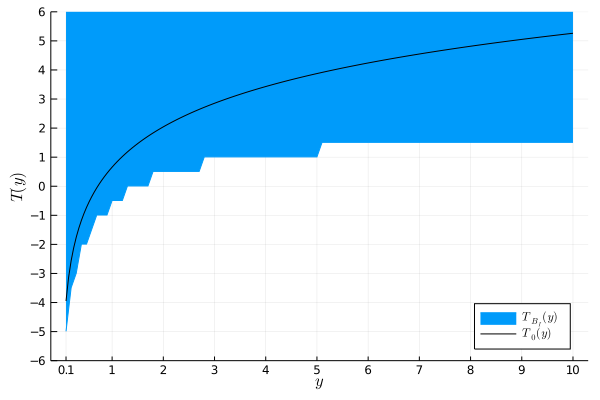}
\end{center}
\end{minipage}
\begin{minipage}[t]{0.33\hsize}
\begin{center}
{\small (c-3) Model 3 \& Support (iii)}
\includegraphics[scale=0.27]{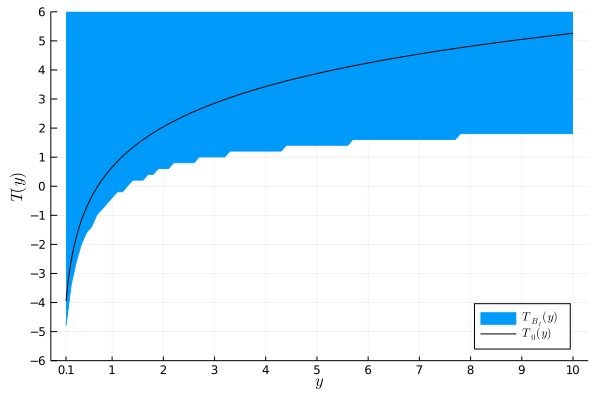}
\end{center}
\end{minipage}
\end{tabular}
\begin{tablenotes}\footnotesize
\item[] Notes: The shaded area in each figure is the computed $T_{B_I}(y)$ for $0.1 \leq y \leq 10$. The solid line in each figure is the graph of the true normalized transformation function $T_0(y)=2(\log y - \log 0.77)$.
\end{tablenotes} \label{fig:id set for T}
\end{figure}

\bigskip


\subsection{Monte Carlo Experiments \label{sec:monte_carlo}}
This subsection presents Monte Carlo simulation results to evaluate the finite sample size and power properties of the inference method for the regression parameters presented in Section \ref{sec:inference}. 
We here use two DGPs (DGP1 and DGP2). In DGP1 and DGP2, the data are derived from Models 2 and 3, respectively, where $X_{1}$ is distributed as $N(0,2)$; $X_{2}$ is uniformly distributed over $\left\{ 0,1\right\}$; $U$, $V$, and $W$ have independent unit exponential distributions. The censoring rates in DGP1 and DGP2 are approximately $16\%$ and $30\%$, respectively.
\par
For the Monte Carlo experiments, 500 samples are randomly drawn with sample sizes $n=100$, $250$, and $500$. Based on the inference method, we conduct a test of $H_{0}:$ (\ref{eq:conditional moment inequality}) holds against $H_{1}:$ (\ref{eq:conditional moment inequality}) is violated at each value of $\left(\beta_{1},\beta_{2}\right)\in\left\{ 1\right\} \times\left\{ -1,-0.5,\ldots,9\right\} $. The true  value of the normalized ${\beta}_{2}$ is 3. 
Critical values are simulated using 1,000 repetitions for the significance level $\alpha=0.05$.
As a baseline case, we set the tuning parameters in the test statistics and GMS function to $R=5$, $\epsilon = 0.0001$, $B_{n}=B_{n}^{bc}\equiv\left(0.8\ln\left(n\right)/\ln\ln\left(n\right)\right)^{\frac{1}{2}}$, and $\kappa_{n}=\kappa_{n}^{bc}\equiv\left(\left(1-\hat{p}_{1-D}^{1/3}\right)^{2/5}\times0.6\ln(n)\right)^{\frac{1}{2}}$.
We set $\eta=10^{-6}$ throughout this and the following sections.
We do not assume that the researcher knows the exact distribution of $X_{1}$; hence, we transform $X_{1}$ into $\tilde{X}_{1}$ as described in Section \ref{sec:test statistics} and use $\tilde{X}_{1}$ instead of $X_{1}$. We assume that the researcher knows the support of $X_2$. Thus we apply the inference method presented in Section \ref{sec:inference} with $\MX=\MX_{1} \times \MX_{2} = [0,1]\times \{0,1\}$. 
\par
Figure \ref{fig:rejection frequency} shows the graphs of rejection frequencies for DGP1 and DGP2 given the baseline case of the tuning parameters with $\epsilon=0.0001$ or $\epsilon = 0.001$. All the rejection frequencies at the true value of normalized $\beta_2$ are close to or lower than the nominal size $\alpha=0.05$ (though the rejection frequency substantially exceeds 0.05 in DGP1 with $n=100$ and $\epsilon=0.0001$). The rejection frequencies are also close to or lower than 0.05 in the corresponding intervals computed in Table \ref{tab:numerical example} with Support (iii). When $\epsilon=0.001$, the test has less power but is more likely to be conservative even in small sample. The power of the test is higher at small values of $\beta_{2}$ than at large values.
\par
Table \ref{tab:robustness check} shows the rejection frequencies at the true parameter value and at $\left(\beta_{1},\beta_{2}\right)=\left(1,0\right)$ for several choices of the tuning parameters in DGP1 and DGP2. The point $\left(\beta_{1},\beta_{2}\right)=\left(1,0\right)$ is not contained in $B_{I}$, as seen from the results of the numerical examples. Table \ref{tab:robustness check} shows the degree of sensitivity of the test to variation in the sample size $n$, the choice of the truncation integer $R$ in the approximate test statistic, the value of $\epsilon$ for the modified variance estimator $\bar{\sigma}_{n}^{2}\left(\beta,g\right)$, and the choice of $B_n$ and $\kappa_n$ in the GMS function. In the case of the last row of Table \ref{tab:robustness check}, $\kappa_{n}$ does not depend on the sample censoring rate $\hat{p}_{1-D}$. The results show that there is some sensitivity to the sample size, the choices of $\kappa_n$ and $R$, and the value of $\epsilon$. In particular, the sensitivity to the choice of $\kappa_{n}$ is high. A small value of $\epsilon$ leads to the test having high power, but can lead to the test having incorrect size in small samples. 
\par
The supplementary material contains two additional Monte Carlo simulation results. One examines inference on the regression parameters in models with three covariates. The other evaluates the performance of a joint inference method for the regression parameters and transformation function introduced in Appendix \ref{app:joint_inference}.


\bigskip
\begin{figure}[h]
\caption{Rejection Frequencies in DGP1 and DGP2}
\begin{tabular}{cc}
\begin{minipage}[t]{0.5\hsize}
\begin{center}
{\small (a) DGP1 \& $\epsilon = 0.0001$}
\includegraphics[scale=0.39]{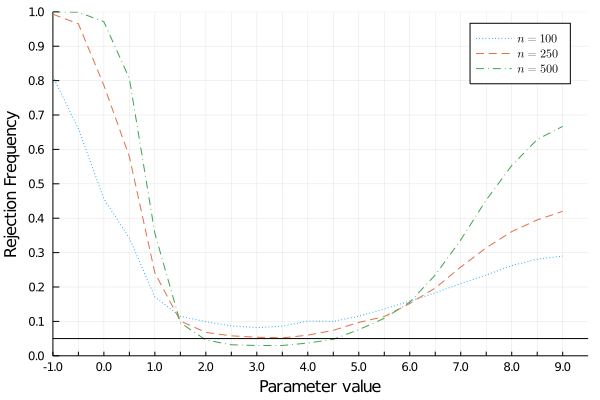}
\end{center}
\end{minipage}
\begin{minipage}[t]{0.5\hsize}
\begin{center}
{\small (b) DGP1 \& $\epsilon = 0.001$}
\includegraphics[scale=0.39]{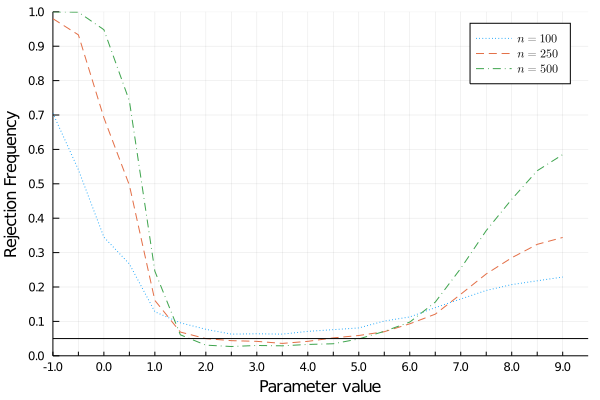}
\end{center}
\end{minipage}
\bigskip
\\
\begin{minipage}[c]{0.5\hsize}
\begin{center}
{\small (c) DGP2 \& $\epsilon = 0.0001$}
\includegraphics[scale=0.39]{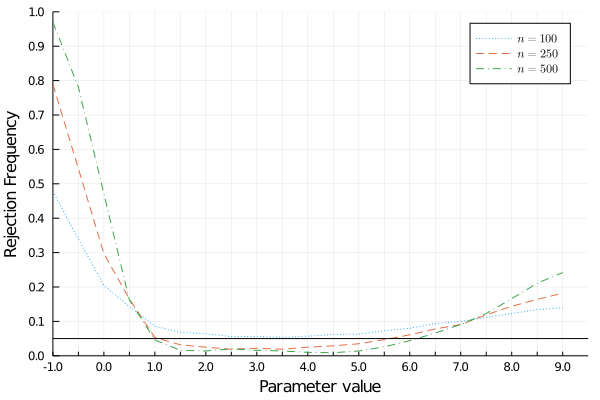}
\end{center}
\end{minipage}
\begin{minipage}[c]{0.5\hsize}
\begin{center}
{\small (d) DGP2 \& $\epsilon = 0.001$}
\includegraphics[scale=0.39]{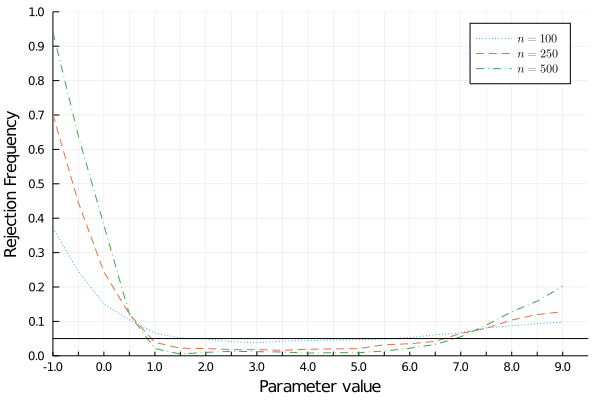}
\end{center}
\end{minipage}
\end{tabular}
\begin{tablenotes}\footnotesize
\item[] Notes: The dashed, dotted, dash-dotted lines in each figure indicate rejection frequencies for sample sizes of $100$, $250$, and $500$, respectively. The solid horizontal line in each figure indicates the significance level of 0.05. 
\end{tablenotes} \label{fig:rejection frequency}
\end{figure}

\bigskip
\begin{table}[h]
\caption{Rejection Frequencies for the Inference Method: Variation in Sample
Size and Choice of the Tuning Parameters}
\vspace*{-0.5cm}
\begin{center}
\scalebox{0.85}[0.85]{
\begin{tabular}{cccccc}
\hline 
 & \multicolumn{2}{c}{DGP1} && \multicolumn{2}{c}{DGP2} \tabularnewline
\cline{2-3} \cline{5-6} 
Case & {\small $\left(\beta_{1},\beta_{2}\right)=\left(1,3\right)$} & {\small $\left(\beta_{1},\beta_{2}\right)=\left(1,0\right)$}  & &
{\small $\left(\beta_{1},\beta_{2}\right)=\left(1,3\right)$}  & {\small $\left(\beta_{1},\beta_{2}\right)=\left(1,0\right)$}  \tabularnewline
\hline 
Baseline Case                 & \multirow{2}{*}{0.054} & \multirow{2}{*}{0.786} && \multirow{2}{*}{0.022} & \multirow{2}{*}{0.297}\tabularnewline
$((n,R,B_n,\kappa_n,\epsilon)=(250,5,B_{n}^{bc},\kappa_{n}^{bc},0.0001))$ &&&&&\tabularnewline
$\epsilon=0.001$          & 0.042 & 0.692 && 0.018 & 0.244\tabularnewline
$\epsilon=0.00001$        & 0.061 & 0.805 && 0.022 & 0.046\tabularnewline
$R=3$                     & 0.052 & 0.594 && 0.046 & 0.192\tabularnewline
$R=7$                     & 0.033 & 0.824 && 0.019& 0.322 \tabularnewline
$n=100,\epsilon=0.001$    & 0.064 & 0.345 && 0.038 & 0.150 \tabularnewline
$n=100,\epsilon=0.0001$   & 0.082 & 0.456 && 0.056 & 0.205 \tabularnewline
$n=100,\epsilon=0.00001$  & 0.089 & 0.489 && 0.058& 0.227 \tabularnewline
$n=500,\epsilon=0.001$    & 0.030 & 0.948 && 0.012 & 0.380 \tabularnewline
$n=500,\epsilon=0.0001$   & 0.030 & 0.971 && 0.016& 0.471 \tabularnewline
$n=500,\epsilon=0.00001$  & 0.032 & 0.970 && 0.018 & 0.492 \tabularnewline
$n=1000,\epsilon=0.001$   & 0.021 & 1.000 && 0.004 & 0.543 \tabularnewline
$n=1000,\epsilon=0.0001$  & 0.031 & 1.000 && 0.010 & 0.675 \tabularnewline
$n=1000,\epsilon=0.00001$ & 0.037 & 1.000 && 0.012 &0.694 \tabularnewline
$B_n=B_n^{bc} /2$ & 0.053 & 0.786 && 0.021& 0.297 \tabularnewline
$B_n=2B_n^{bc}$ & 0.054 & 0.786 && 0.022& 0.297 \tabularnewline
$\kappa_n=\kappa_n^{bc} /2$ & 0.293 & 0.944 && 0.148& 0.552 \tabularnewline
$\kappa_n=2\kappa_n^{bc}$ & 0.000 & 0.243 && 0.000& 0.045 \tabularnewline
$\left(B_{n},\kappa_{n}\right)=1/2\left(B_{n}^{bc},\kappa_{n}^{bc}\right)$ & 0.288 & 0.944 && 0.140& 0.550 \tabularnewline
$\left(B_{n},\kappa_{n}\right)=2\left(B_{n}^{bc},\kappa_{n}^{bc}\right)$ & 0.000 & 0.243 && 0.000& 0.045 \tabularnewline
$\kappa_{n}=\left(0.6\ln(n)\right)^{\frac{1}{2}}$ & 0.026 & 0.683 && 0.010& 0.217 \tabularnewline
\hline 
\end{tabular}
}
\label{tab:robustness check}
\end{center}
\end{table}
\bigskip


\section{Empirical Illustration \label{sec:empirical illustration}}

We apply the proposed inference method to evaluate the effect of heart transplants on patients' survival duration using the Stanford Heart Transplant Data taken from \cite{Kalbfleisch_Prentice_1980}. The data set consists of survival times (in days) of 103 patients; an indicator of censoring, which takes the value one if the patient is dead (uncensored) or zero if the patient is censored; an indicator of receiving a heart transplant, which takes the value one if the patient receives a heart transplant or zero otherwise; and the age (in years) of patients at the time of acceptance into the program. Among the 103 patients, $27\%$ (28 patients) are censored due to attrition or administrative censoring. The censoring rates for the treated (receive a transplant) and
untreated (do not receive a transplant) groups are $35\%$ and $22\%$, respectively.
\par
We consider the following censored transformation model,
\begin{align*}
T\left(Y_{0i}\right) =  \min\left\{\beta_{cons} +  X_{i,age}\beta_{age}+X_{i,treat}\beta_{treat}+U_{i},T\left(C_{i}\right)\right\} ,
\end{align*}
where, for each patient $i$, $Y_{0i}$ is the observed survival time, $X_{i,age}$ is the age, $X_{i,treat}$ is the transplant indicator, $U_{i}$ is unobserved heterogeneity, $C_{i}$ is the censoring time, and $T$ is a strictly increasing function.
Applying the proposed method, we allow the censoring to be arbitrarily correlated with the patient's age and unobserved heterogeneity. Furthermore, we do not specify the transformation function or the distribution function of the patient's unobserved heterogeneity. For scale normalization, we set $\left|\beta_{age}\right|=1$; for location normalization, we set $\beta_{cons}=0$ and $T(\tilde{y})=0$ with $\tilde{y}=90$ being the median of $Y_{0}$ in the sample.
Our interest is in the normalized regression parameter $\beta_{treat}$ and the normalized transformation function $T$.
We compare the proposed method with the partial rank estimator (PRE) proposed by \cite{Khan_Tamer_2007}. This estimator is robust up to covariate-dependent censoring and consistently estimates the normalized regression parameters in the nonparametric transformation model.
\par
Table \ref{tab:empirical illustration} shows the inference results for $\beta_{treat}$. It presents the point estimate obtained from the PRE and 95\% confidence intervals obtained from the PRE and the proposed method. The confidence interval obtained from the PRE is computed based on 1,000 bootstrap pseudo samples from the data. For the proposed method, we set the tuning parameters to $R=5$, $B_{n}=\left(0.8\ln\left(n\right)/\ln\ln\left(n\right)\right)^{\frac{1}{2}}$, $\kappa_{n}=\left(\left(1-\hat{p}_{1-D}^{1/3}\right)^{2/5}\times0.6\ln(n)\right)^{\frac{1}{2}}$, and $\eta=10^{-6}$ as in the baseline case in the Monte Carlo simulation in the previous section. For $\epsilon$, we use both $\epsilon=0.0001$ and $\epsilon=0.001$. Setting $\epsilon=0.001$ is more conservative in a small sample like the Stanford Heart Transplant data set according to the Monte Carlo simulation results in the previous section. The confidence intervals obtained from the proposed method do not have finite upper bounds. This would be because age does not have sufficiently large support to derive a finite upper bound on $\beta_{treat}$ in $B_I$. The estimate obtained from the PRE is positive and is significantly different from zero. The 95\% confidence interval obtained from the proposed method is also entirely positive regardless of the choice of $\epsilon$. With the conservative choice of $\epsilon=0.001$, the $95\%$ confidence interval obtained from the proposed method covers the confidence interval obtained from the PRE. The inference results relating to the proposed method show that even if censoring is arbitrarily correlated with a patient's age or unobserved heterogeneity, a heart transplant has a positive effect on the patient's survival time. 
\par
\bigskip
\begin{table}[h]
\caption{Empirical Results for $\beta_{treat}$}
\vspace*{-0.5cm}
\begin{center}
\begin{tabular}{cccc}
\hline 
 &  \multirow{2}{*}{PRE} & \multicolumn{2}{c}{Proposed Method}  \tabularnewline \cline{3-4}
 &   & $\epsilon=0.001$ & $\epsilon=0.0001$ \tabularnewline
\hline 
Estimate & 42.6 & - & -\tabularnewline
95\% Confidence Interval & {[}17.2, 57.3{]} & {[}10.4, +$\infty${]}& {[}31.3, +$\infty${]} \tabularnewline
\hline 
\end{tabular}
\label{tab:empirical illustration}
\end{center}
\end{table}

We next apply the joint inference procedure for the regression parameters and transformation function presented in Appendix \ref{app:joint_inference}. Figure \ref{fig:empirics_CI of T} shows the marginal $95\%$-confidence set of $T(y)$ at each $y \in [0,900]$, which is the projection of the three-dimensional confidence interval of $(\beta_{age},\beta_{treat},T(y))$ on the one-dimension.
The same tuning parameters as those used in Table \ref{tab:empirical illustration} are used. Since $T_{I,\beta}(y)$ does not have a finite upper bound, neither do the confidence intervals. The estimated confidence intervals also do not have finite lower bounds for small values of $y$. The $95\%$-confidence lower bound on $T(y)$ increases rapidly with $y$ in the approximate interval $[100,250]$, but it changes slightly when $y$ is large.

\bigskip
\begin{figure}[h]
\caption{Empirical Results for the Transformation Function}
\begin{tabular}{cc}
\begin{minipage}[t]{0.5\hsize}
\begin{center}
{\small (a) $\epsilon = 0.001$}
\includegraphics[scale=0.4]{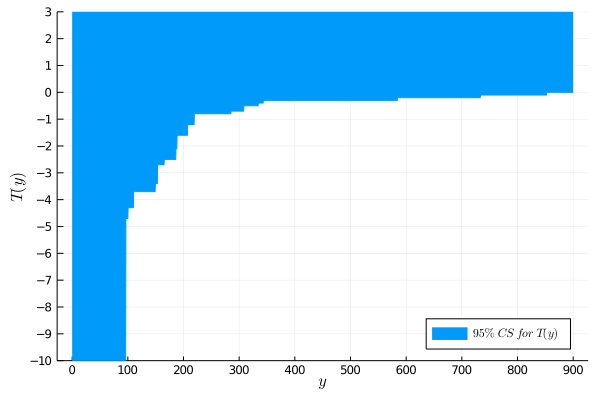}
\end{center}
\end{minipage}
\begin{minipage}[t]{0.5\hsize}
\begin{center}
{\small (b) $\epsilon = 0.0001$}
\includegraphics[scale=0.4]{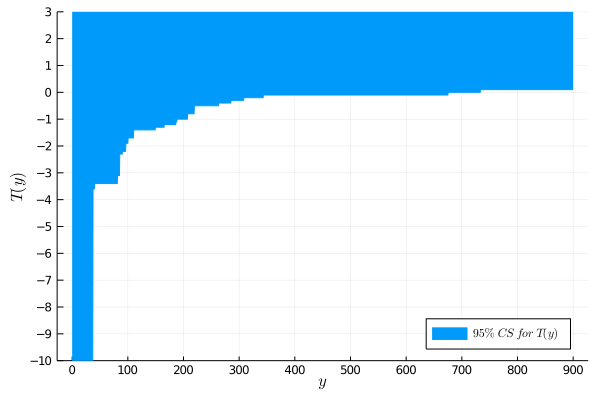}
\end{center}
\end{minipage}
\end{tabular}
\begin{tablenotes}\footnotesize
\item[] Notes: 
The shaded are in each figure is the computed confidence interval of $T(y)$ at each $y \in [0,900]$. The confidence intervals are computed on the range $[-10,3]$.
\end{tablenotes} \label{fig:empirics_CI of T}
\end{figure}

\section{Concluding Remarks \label{sec:conclusion}}

In this paper, we propose a partial identification and inference approach for a nonparametric transformation model in the presence of endogenous censoring. We develop bounds on the regression parameters and the transformation function, each of which is characterized by conditional moment inequalities involving U-statistics. We also characterize the sharp identified set of the regression parameters, using concepts from random set theory, though this set is hard to compute. A comparison of the proposed set and the sharp identified set characterization makes it clear when the proposed set approaches the sharp set. Based on the identification result, we propose an inference method for the regression parameters by extending the inference approach for conditional moment inequality models, proposed by \cite{Andrews_Shi_2013}, to the U-statistics case. We also derive the asymptotic properties of this approach. A joint inference procedure for the regression parameters and transformation function is presented in Appendix \ref{app:joint_inference}.
Numerical examples illustrate the characteristics of the proposed sets for the regression parameters and transformation function, and the results of Monte Carlo experiments demonstrate the size and power properties of the proposed inference methods. As an empirical application, we apply the inference methods to evaluate the effect of heart transplants on patients' survival duration using data from the Stanford Heart Transplant Study, for which we find that heart transplants have a positive effect on patients' survival duration regardless of the censoring mechanism.


\appendix

\part*{Appendix \label{app:appendix}}

In this appendix, Section \ref{app:proof_1} provides proofs of Theorems \ref{thm:partial identification} and \ref{thm:transformation function}. Section \ref{app:proof_2} provides a proof of Theorem \ref{thm:asymptotic property} with some auxiliary lemmas. Section \ref{app:joint_inference} presents a joint inference procedure for the regression parameters and the transformation function based on the identification results presented in Sections \ref{sec:partial identification} and \ref{sec:transformation function}. Some additional numerical studies and Monte Carlo simulation studies are presented in the supplementary material to this paper.

\section{Proofs of Theorems \ref{thm:partial identification} and \ref{thm:transformation function} \label{app:proof_1}}
This section provides proofs of Theorems \ref{thm:partial identification} and \ref{thm:transformation function}.

\begin{proof}[Proof of Theorem \ref{thm:partial identification}]
From the definitions of $Y_{1i}$ and $Y_{0i}$, the following holds for all $(x_{i},x_{j})\in\mathcal{X}^{2}$,
\begin{align*}
P(Y_{1i}\geq Y_{0j}\mid x_{i},x_{j}) \geq  P(Y_{i}^{\ast}\geq Y_{j}^{\ast}\mid x_{i},x_{j}).
\end{align*}
For the conditional multinomial distribution $P(Y_{i}^{\ast},Y_{j}^{\ast}\mid x_{i},x_{j})$,
it follows that
\begin{align*}
P(Y_{i}^{\ast}\geq Y_{j}^{\ast} \mid x_i,x_j) &= P(\Lambda(x_{i}^{\prime} \beta_0 + U_i) \geq \Lambda(x_{j}^{\prime} \beta_0 + U_j)) \\
&\geq P(x_{i}^{\prime} \beta_0 + U_i \geq x_{j}^{\prime} \beta_0 + U_j)\\
&= P(U_i - U_j \geq - (x_i - x_j)^{\prime} \beta_0 ). 
\end{align*}
The first line follows from Assumptions \ref{asm:iid} and \ref{asm:independence}. The second line follows because $\Lambda$ is a monotonically increasing function by Assumption \ref{asm:monotonicity}. When $(x_i - x_j)^{\prime} \beta_0 \geq 0$,
\begin{align*}
    P(U_i - U_j \geq -(x_i - x_j)^{\prime}\beta_0 )  \geq \frac{1}{2}
\end{align*}
holds because $U_i -U_j$ is symmetrically distributed about zero under Assumption \ref{asm:iid}. Therefore it follows that
\begin{align*}
x_{i}^{\prime}\beta_{0}\geq x_{j}^{\prime}\beta_{0} & \Rightarrow  P\left(Y_{i}^{\ast}\geq Y_{j}^{\ast}\mid x_{i},x_{j}\right)\geq \frac{1}{2}\\
 & \Rightarrow  P(Y_{1i}\geq Y_{0j}\mid x_{i},x_{j})\geq \frac{1}{2}
\end{align*}
for all $(x_{i},x_{j})\in\mathcal{X}^{2}$. This implies that $\beta_{0}\in B_{I}$ a.s.. 

We next show that $B_I$ is a proper subset of $B$. Let $(\tilde{x}_i,\tilde{x}_j) \in \widetilde{\MX}^2$. Then there exist $\beta \in B$ such that
\begin{align*}
    \tilde{x}_{i}^{\prime}\beta \geq \tilde{x}_{j}^{\prime}\beta \Rightarrow P(Y_{1i}\geq Y_{0j} \mid \tilde{x}_i,\tilde{x}_j) <\frac{1}{2}.
\end{align*}
Such a $\beta$ is not a member of $B_I$. Hence, because $P((X_i,X_j) \in \widetilde{\MX}^2)>0$ by Assumption \ref{asm:support condition}, $B_I \subset B$ holds a.s.
\end{proof}

\bigskip

\begin{proof}[Proof of Theorem \ref{thm:transformation function}]
Under Assumptions \ref{asm:iid}, \ref{asm:independence}, and \ref{asm:strict monotonicity}, $\beta_{0} \in B_{I}$ a.s. from Theorem \ref{thm:partial identification}. Then it suffices to show that $T\left(y\right) \in T_{I,\beta_{0}}\left(y\right)$ a.s. for any $y\in\mathbb{R}$.
\par
Let $y \in \Real$ be fixed. Note that from the definitions of $Y_{1i}$ and $Y_{0i}$,
\begin{align*}
P(Y_{1i}\geq y\mid x_{i})-P(Y_{0j}\geq\tilde{y}\mid x_{j})  \geq P(Y_{i}^{\ast}\geq y\mid x_{i})-P(Y_{j}^{\ast}\geq\tilde{y}\mid x_{j}).
\end{align*}
holds for all $(x_{i},x_{j})\in\mathcal{X}^{2}$.
\par
For the conditional multinomial distribution $P(Y_{i}^{\ast},Y_{j}^{\ast}\mid x_{i},x_{j})$, it follows that
\begin{align*}
 & P(Y_{i}^{\ast}\geq y\mid x_{i})-P(Y_{j}^{\ast}\geq\tilde{y}\mid x_{j})\\
=&\ P\left(x_{i}^{\prime}\beta_{0}+U_{i}\geq T\left(y\right)\mid x_{i}\right)-P\left(x_{j}^{\prime}\beta_{0}+U_{j}\geq0\mid x_{j}\right)\\
= &\ F_{U}\left(-x_{j}^{\prime}\beta_{0}\right)-F_{U}\left(T\left(y\right)-x_{i}^{\prime}\beta_{0}\right).
\end{align*}
The first equality follows from Assumptions \ref{asm:iid} and \ref{asm:strict monotonicity} and the fact that $T\left(\tilde{y}\right)=0$ for the location normalization. The second equality follows from Assumptions \ref{asm:iid} and \ref{asm:independence}. Since 
\begin{align*}
    x_{i}^{\prime}\beta_{0}-x_{j}^{\prime}\beta_{0}\geq T\left(y\right) \Rightarrow
    F_{U}\left(-x_{j}^{\prime}\beta_{0}\right) \geq F_{U}\left(T\left(y\right)-x_{i}^{\prime}\beta_{0}\right),
\end{align*}
we have
\begin{align*}
x_{i}^{\prime}\beta_{0}-x_{j}^{\prime}\beta_{0}\geq T\left(y\right) & \Rightarrow  P(Y_{i}^{\ast}\geq y\mid x_{i})-P(Y_{j}^{\ast}\geq\tilde{y}\mid x_{j})\geq0\\
 & \Rightarrow  P(Y_{1i}\geq y\mid x_{i})-P(Y_{0j}\geq\tilde{y}\mid x_{j})\geq0
\end{align*}
for all $(x_{i},x_{j})\in\mathcal{X}^{2}$. This implies that $T\left(y\right)\in T_{I,\beta_{0}}\left(y\right)$ a.s.
\end{proof}

\section{Proof of Theorem \ref{thm:asymptotic property} \label{app:proof_2}}

This section provides a proof of the uniform asymptotic probability results for the inference method presented in Section \ref{sec:inference}. The outline of the proof is same as that of the proofs of Theorems 2(b) and 3 in AS, but we modify them for the case of U-statistics. Let $\rightsquigarrow$ denote weak convergence of a stochastic process in the sense of \cite{Pollard_1990}. The following notations are similar to the notations introduced in AS, 
\begin{align*}
v_{n,P}\left(\beta,g\right)\equiv n^{\frac{1}{2}}\left(\bar{m}_{n}\left(\beta,g\right)-E_{P}\left[m(W_{i},W_{j},\beta,g)\right]\right)/\sigma_{p}\left(\beta\right)
\end{align*}
and
\begin{align*}
\hat{h}_{2,n,P}\left(\beta,g,g^{\ast}\right)  \equiv & \left\{ \frac{1}{n\left(n-1\right)\left(n-2\right)}\sum_{i\neq j\neq k}m(W_{i},W_{j},\beta,g)m(W_{i},W_{k},\beta,g^{\ast})\right.\\
   & \left.-E_{P}\left[m(W_{i},W_{j},\beta,g)\right]\cdot E_{P}\left[m\left(W_{i},W_{j},\beta,g^{\ast}\right)\right]\right\} /\sigma_{P}^{2}\left(\beta\right).
\end{align*}
Let $\left\{ v_{h_{2}}\left(g\right):g\in{\cal G}\right\} $ be a mean zero Gaussian process with some covariance kernel $h_{2}\left(\cdot,\cdot\right)$ on ${\cal G}\times{\cal G}$. Let $\{a_{n}:n\geq 1\}$ denote a subsequence of $\{n\}$.
\par
To prove Theorem \ref{thm:asymptotic property}, we first prove that the following two lemmas hold. Lemma \ref{lem:lemma_1} implies that Assumption EP in AS holds. Lemmas \ref{lem:lemma_2} implies that a version of Assumption CI in AS, which is modified for the case of U-statistics, holds.
\bigskip

\begin{lemma} \label{lem:lemma_1}
For any subsequence $\left\{ \left(\beta_{a_{n}},P_{a_{n}}\right)\in{\cal Q}:n\geq1\right\} $
such that
\begin{align*}
\underset{n\rightarrow\infty}{\lim}\underset{g,g^{\ast}\in{\cal G}}{\sup}\left\Vert h_{2,P_{a_{n}}}\left(\beta_{a_{n}},g,g^{\ast}\right)-h_{2}\left(g,g^{\ast}\right)\right\Vert =0
\end{align*}
for some covariance kernel $h_{2}(\cdot,\cdot)$ on ${\cal G}\times{\cal G}$,
we have\\
(a) $\sqrt{a_{n}}v_{a_{n},P_{a_{n}}}\left(\beta_{a_{n}},\cdot\right)\rightsquigarrow v_{h_{2}}\left(\cdot\right)$
as $n\rightarrow\infty$, and \\
(b) $\underset{\left(g,g^{\ast}\right)\in{\cal G}\times\mathcal{G}}{\sup}\left\Vert \hat{h}_{2,a_{n},P_{a_{n}}}\left(\beta_{a_{n}},g,g^{\ast}\right)-h_{2}\left(g,g^{\ast}\right)\right\Vert \underset{p}{\rightarrow}0$
as $n\rightarrow\infty$.
\end{lemma}

\bigskip

\begin{lemma} \label{lem:lemma_2}
For any $\beta\in B$ and any distribution function $P$ that satisfies Assumptions \ref{asm:monotonicity}--\ref{asm:independence}, let 
\begin{align*}
{\cal X}_{P}\left(\beta\right) = \left\{ (x_{i},x_{j})\in{\cal X}^{2}:E_{P}\left[m\left(W_{i},W_{j},\beta\right)\mid x_{i},x_{j}\right]<0\right\} .
\end{align*}
Then, for any $\beta\in B$ and $P$ that satisfies Assumptions \ref{asm:monotonicity}--\ref{asm:independence} and $P\left((x_{i},x_{j})\in{\cal X}_{P}\left(\beta\right)\right)>0,$
there exists some $g\in{\cal G}$ such that
\begin{align*}
E_{P}\left[m(W_{i},W_{j},\beta,g)\right]<0.
\end{align*}
\end{lemma}

\bigskip

The next lemma with a proof is auxiliary to Lemma \ref{lem:lemma_1}.
\bigskip

\begin{lemma} \label{lem:lemma_3}
Let $(\beta,P) \in \mathcal{{Q}}$. Define classes of functions ${\cal F}_{1}=\left\{ f_{1}(\cdot,\cdot,\beta,g): g\in{\cal G}\right\}$ and ${\cal F}_{2}=\left\{ f_{2}(\cdot,\cdot,\cdot,\beta,g,g^{\ast}):\left(g,g^{\ast}\right)\in{\cal G}\times{\cal G}\right\}$, where
\begin{align*}
f_{1}(w_{i},w_{j},\beta,g)=m(w_{i},w_{j},\beta,g)-E_{P}\left[m(W_{i},W_{j},\beta,g)\right]
\end{align*}
and 
\begin{align*}
f_{2}(w_{i},w_{j},w_{k},\beta,g,g^{\ast}) =&\  m(w_{i},w_{j},\beta,g)\cdot m(w_{i},w_{k},\beta,g^{\ast})\\
& -E_{P}\left[m(W_{i},W_{j},\beta,g)\right]\cdot E_{P}\left[m\left(W_{i},W_{j},\beta,g^{\ast}\right)\right]. 
\end{align*}
Then $\mathcal{F}_{1}$ and $\mathcal{F}_{2}$ are Euclidean classes of functions for constant envelopes $1$ and $1/2$, respectively, in the sense of \cite{Nolan_Pollard_1987}.
\end{lemma}

\begin{proof}[Proof of Lemma \ref{lem:lemma_3}]
We consider the class of function $\mathcal{G}$ defined in Section \ref{sec:test statistics}. 
For any $x \in \Real^{k}$, let $x=(x^{(1) \prime},x^{(2) \prime})$ where $x^{(1)}$ is a vector of the first $p$ elements of $x$ and $x^{(2)}$ is a vector of the remaining $k-p$ elements of $x$. 
$\MG$ is represented as
\begin{align*}
\mathcal{G}  = & \left\{I\left[\frac{a-\mathbf{1}_{p}}{2r}<x_{i}^{(1)}\leq \frac{a}{2r} , x_{i}^{(2)}=b\right]\cdot
I\left[\frac{\tilde{a}-\mathbf{1}_{p}}{2r}<x_{j}^{(1)}\leq \frac{\tilde{a}}{2r}, x_{j}^{(2)}=\tilde{b}\right]:\right.\\
 & a = (a_{1},\ldots,a_{p})^{\prime},\ \tilde{a} =\left(\tilde{a}_{1},\ldots,\tilde{a}_{p}\right)^{\prime},\ \left(a_{u},\tilde{a}_{u}\right)\in\left\{ 1,2,\ldots,2r\right\} ^{2}\\
& \left. \mbox{for }u=1,\ldots,p\ \mbox{and }r=1,2,\ldots , \mbox{ and }
(b, \tilde{b}) \in \MX_{2}^{2}
\right\},
\end{align*}
where $\mathbf{1}_{p}$ is a $p$-dimensional vector of ones.
Because the collection of cells in $\Real^{2k}$ is a Vapnik-Chervonenkis (VC) class of sets (see \citeauthor{van_der_vaart_Wellner_1996} (\citeyear{van_der_vaart_Wellner_1996}, Example 2.6.1)), the collection of all subgraphs, $\left\{ \left(x_{i},x_{j},t\right):t<g(x_{i},x_{j})\right\} $, of the function in $\mathcal{G}$ forms a VC class of sets in $\mathcal{X}^{2}\times\mathbb{R}$. Hence, $\mathcal{G}$ is a VC class of functions. Combining this result with Lemma 2.6.18 in \cite{van_der_vaart_Wellner_1996}, $\mathcal{F}_{1}$ and $\mathcal{F}_{2}$ are VC-classes of functions. Thus, from Corollary 19 in \cite{Nolan_Pollard_1987}, $\mathcal{F}_{1}$ and $\mathcal{F}_{2}$ are Euclidean classes of functions. $\mathcal{F}_{1}$ and $\mathcal{F}_{2}$ obviously have the constant envelopes $1$ and $1/2$, respectively, from their definitions.
\end{proof}

\bigskip

We provide proofs of Lemmas \ref{lem:lemma_1} and \ref{lem:lemma_2} and Theorem \ref{thm:asymptotic property} below. 
\bigskip

\begin{proof}[Proof of Lemma \ref{lem:lemma_1}.(a)]
While Lemma \ref{lem:lemma_1} is stated in terms of a subsequence $\left\{ a_{n}\right\}$, for notational simplicity, we prove it for the sequence $\left\{ n\right\}$. All of the arguments in this and the next proofs proceed with $\left\{ a_{n}\right\} $ instead of $\left\{ n\right\}$.
\par
We use Theorem 5 in \cite{Nolan_Pollard_1988} to show that the weak convergence result in Lemma \ref{lem:lemma_1}.(a) holds. Let $(\beta,P)$ be the limit of $(\beta_{n},P_{n})$ and $h_{2}(\cdot,\cdot)=h_{2,P}(\beta,\cdot,\cdot)$. Let $N_{p}\left(\epsilon,Q,{\cal F},F\right)$ denote the $L_{p}\left(Q\right)$-covering number of radius $\epsilon$ for a functional space ${\cal F}$ with a envelope function $F$ where $Q$ is some probability measure on $\MX$. For any distribution $Q$ on $\MX$, let $Q{\cal F}_{1} \equiv \{Qf_{1}(x,\cdot,\beta,g): g \in \MG\}$ be the class of functions $Qf_{1}\left(x,\cdot,\beta,g\right)$ on ${\cal X}$, where $f_{1}$ and $\mathcal{F}_{1}$ are defined in Lemma \ref{lem:lemma_3} and $Qf_{1}(x,\cdot,\beta,g)\equiv \int_{\MX} f_{1}(x,x_2,\beta,g)dQ_{X}(x_2)$. 
\par
To apply Theorem 5 in \cite{Nolan_Pollard_1988}, it suffices to show that the following conditions hold:
\begin{description}
\item [{(i)}] $\sup_{Q}\intop_{0}^{1}\log N_{2}\left(\epsilon,Q\times Q,{\cal F}_{1},F\right)d\epsilon<\infty$, $\sup_{Q}\left(\intop_{0}^{1}\log N_{2}\left(\epsilon,Q\times Q,\MF_{1},F\right)d\epsilon\right)^{2}<\infty$, and
$\sup_{Q}\left(\intop_{0}^{1}\log N_{2}\left(\epsilon,Q,P\MF_{1},PF\right)d\epsilon\right)^{2}<\infty$;
\item [{(ii)}] as $\gamma\searrow 0$, 
\begin{align*}
\sup_{Q}\intop_{0}^{\gamma}\log N_{2}\left(\epsilon,Q,P{\cal F}_{1},PF\right) \rightarrow 0,
\end{align*}
\end{description}
where $Q$ is any probability measure on $\MX$. We show below that these two conditions are satisfied. 
\par
We first consider condition (i). From Lemma \ref{lem:lemma_3}, the class of functions $\mathcal{F}_{1}$ is Euclidean with the constant envelope $F=1$. Then, from Corollary 21 in \cite{Nolan_Pollard_1987}, the class of functions $P{\cal F}_{1}$ is also a Euclidean class with the constant envelope $1$. From page 789 in \cite{Nolan_Pollard_1987}, if a Euclidean class has a constant envelope function, then the upper bound on the $L_{p}(Q)$-covering number of radius $\epsilon$ for it is uniform in any probability measure $Q$ on $\MX$. Therefore, since $\MF_{1}$ and $P{\cal F}_{1}$ are Euclidean classes with constant envelopes, for any $0<\epsilon\leq1$, there exist some constants
$K_{2}$, $K_{2}^{\ast}$, $V_{2}$, and $V_{2}^{\ast}$ such that $N_{2}\left(\epsilon,Q\times Q,{\cal F}_{1},F\right) \leq K_{2}\epsilon^{-2V_{2}}$
and $N_{2}\left(\epsilon,Q,P{\cal F}_{1},PF\right) \leq K_{2}^{\ast}\epsilon^{-2V_{2}^{\ast}}$, for any probability measure $Q$ on $\MX$. Then it follows that
\begin{align*}
\sup_{Q}\intop_{0}^{1}\log N_{2}\left(\epsilon,Q\times Q,{\cal F}_{1},F\right)d\epsilon&\leq\intop_{0}^{1}\left(\log K_{2}\epsilon^{-2V_{2}}\right)d\epsilon<\infty,\\
\sup_{Q}\left(\intop_{0}^{1}\log N_{2}\left(\epsilon,Q\times Q,{\cal F}_{1},F\right)d\epsilon\right)^{2}&\leq\left(\intop_{0}^{1}\left(\log K_{2}\epsilon^{-2V_{2}}\right)d\epsilon\right)^{2}<\infty,
\end{align*}
and
\begin{align*}
\sup_{Q}\left(\intop_{0}^{1}\log N_{2}\left(\epsilon,Q,P{\cal F}_{1},PF\right)d\epsilon\right)^{2}\leq\left(\intop_{0}^{1}\left(\log K_{2}^{\ast}\epsilon^{-2V_{2}^{\ast}}\right)d\epsilon\right)^{2}<\infty.
\end{align*}
These imply that condition (i) is satisfied. 
\par
Next, as $\gamma\searrow0$,
\begin{align*}
\sup_{Q}\intop_{0}^{\gamma}\log N_{2}\left(\epsilon,Q,P{\cal F}_{1},PF\right)d\epsilon & \leq \intop_{0}^{\gamma}\log (K_{2}^{\ast}\epsilon^{-2V_{2}^{\ast}})d\epsilon\\
 & = \gamma\log K_{2}^{\ast}-2V_{2}^{\ast}\gamma(\log\gamma - 1)\\
 & \rightarrow 0.
\end{align*}
This implies that condition (ii) is satisfied. Therefore, from Theorem
5 in \cite{Nolan_Pollard_1988} and the fact that $\beta_n \rightarrow \beta$ and $\sigma_{P_{n}}\left(\beta_{n}\right)\underset{p}{\rightarrow}\sigma_{P}\left(\beta\right)$, 
Lemma \ref{lem:lemma_1}.(a) holds by the extended continuous mapping theorem (Theorem 1.11.1 in \cite{van_der_vaart_Wellner_1996}). 
\end{proof}

\bigskip

\begin{proof}[Proof of Lemma \ref{lem:lemma_1}.(b)]
Let $(\beta,P)$ be the limit of $(\beta_{n},P_{n})$. Since $\mathcal{F}_{1}$ and $\mathcal{F}_{2}$ are Euclidean classes with constant envelopes from Lemma \ref{lem:lemma_3}, by applying Corollary 7 in \cite{Sherman_1994}, it follows that
\begin{eqnarray*}
\sup_{g \in {\cal G}} \left\Vert \frac{1}{n\left(n-1\right)}\sum_{i\neq j}f_{1}(W_{i},W_{j},\beta,g)\right\Vert  \underset{p}{\rightarrow}0
\end{eqnarray*}
and
\begin{eqnarray*}
\sup_{(g,g^{*})\in \MG\times \MG} \left\Vert \frac{1}{n\left(n-1\right)\left(n-2\right)}\sum_{i\neq j\neq k}f_{2}(W_{i},W_{j},W_{k},\beta,g,g^{\ast})\right\Vert  \underset{p}{\rightarrow}0.
\end{eqnarray*}
Therefore, letting $h_{2}\left(\beta,g,g^{\ast}\right)=h_{2,P}\left(\beta,g,g^{\ast}\right)$ be given by (\ref{eq:covariance kernel}) and further dividing by $\sigma_{P}^{2}\left(\beta\right)$, as $\beta_{n}\rightarrow \beta$ and $\sigma_{P_{n}}\left(\beta_{n}\right)\underset{p}{\rightarrow}\sigma_{P}\left(\beta\right)$,
we have 
\begin{align*}
 &\ \underset{(g,g^{\ast})\in \MG\times\MG}{\sup}\left\Vert \hat{h}_{2,a_{n},P_{n}}\left(\beta_{n},g,g^{\ast}\right)-h_{2}\left(\beta,g,g^{\ast}\right)\right\Vert \\
\leq &\ \sup_{(g,g^{\ast})\in \MG\times\MG}\left\Vert \frac{1}{n\left(n-1\right)\left(n-2\right)}\sum_{i\neq j\neq k}f_{2}(W_{i},W_{j},W_{k},\beta,g,g^{\ast})/\sigma_{P}\left(\beta\right)\right \Vert \\
 &\ +\left\{ \sup_{g \in{\cal G}}\left\Vert \frac{1}{n\left(n-1\right)\left(n-2\right)}\sum_{i\neq j}f_{1}(W_{i},W_{j},\beta,g)/\sigma_{P}\left(\beta\right)\right \Vert \right\} ^{2}+o_{p}\left(1\right)\\
\underset{p}{\rightarrow} & \ 0.
\end{align*}
\end{proof}

\bigskip

\begin{proof}[Proof of Lemma \ref{lem:lemma_2}]
It suffices to show that
\begin{align}
 & E_{P}\left[m(W_{i},W_{j},\beta,g)\right]\geq0\ \ \forall g\in{\cal G} \notag \\
\Rightarrow\ &E_{P}\left[m\left(W_{i},W_{j},\beta\right)\mid x_{i},x_{j}\right]\geq0\ \ \mbox{for almost every $(x_i,x_j) \in \MX^2$.} \label{eq:giv inequality}
\end{align}
Let $X=(X^{(1)\prime},X^{(2)\prime})^{\prime}$ where $X^{(1)}$ and $X^{(2)}$ are distributed on $\MX_{1}$ and $\MX_{2}$, respectively. Let $(b,\tilde{b}) \in \MX_{2}^{2}$ be fixed, and define a class of instrumental functions
\begin{align*}
\widetilde{\MG}  =  \left\{ g\left(x_{i}^{(1)},x_{j}^{(1)}\right)=I\left[\left(x_{i}^{(1)},x_{j}^{(1)}\right) \in J_{1}\right]\ \mbox{for}\ J_{1}\in\MJ_{1} \right\} ,
\end{align*}
where
\begin{align*}
\mathcal{J}_{1}  \equiv & \left\{ J_{(a,\tilde{a},r)}=\left(\vartimes_{u=1}^{p}\left(\frac{a_{u}-1}{2r},\frac{a_u}{2r} \right]\right) \times \left(\vartimes_{u=1}^{p}\left(\frac{\tilde{a}_{u}-1}{2r},\frac{\tilde{a}_{u}}{2r} \right]\right):\right.\\
   &\ \left. a=(a_{1},\ldots,a_{p}),\ \tilde{a}=(\tilde{a}_{1},\ldots,\tilde{a}_{p}), \ a_u\in\left\{ 1,2,\ldots,2r\right\}, \ \tilde{a}_u\in\left\{ 1,2,\ldots,2r\right\} \right. \\ 
   &\ \left.\mbox{for}\ u=1,\ldots,p\ \mbox{and}\ r=1,2,\ldots \right\} .
\end{align*}
Then, to prove (\ref{eq:giv inequality}), it suffices to show that 
\begin{align}
 & E_{P}\left[m(W_{i},W_{j},\beta,g) \mid (X_{i}^{(2)},X_{j}^{(2)})=(b,\tilde{b}) \right]\geq0\ \ \forall g\in \widetilde{\MG} \notag \\
\Rightarrow\ &E_{P}\left[m\left(W_{i},W_{j},\beta\right)\mid (X_{i}^{(1)},X_{j}^{(1)})=(x_{i}^{(1)},x_{j}^{(1)}), (X_{i}^{(2)},X_{j}^{(2)})=(b,\tilde{b})  \right]\geq0 \notag \\
& \mbox{ for almost every $(x_{i}^{(1)},x_{j}^{(1)}) \in \MX_{1}^{2}$}. 
\label{eq:giv inequality_2}
\end{align}

We invoke Lemma C1 in AS. Let ${\cal R}$ be a semiring of subsets of $\mathbb{R}^{2p}$ and
\begin{align*}
\mu(J_1)  \equiv  E_{P}\left[m\left(W_{i},W_{j},\beta\right)\cdot I\left[\left(X_{i}^{(1)},X_{j}^{(1)}\right)\in J_{1}\right] \mid (X_{i}^{(2)},X_{j}^{(2)})=(b,\tilde{b})\right]
\end{align*}
for $J_{1}\in\sigma\left(\mathcal{J}_{1} \right)={\cal B}\left(\mathbb{R}^{2p}\right)$, where $\sigma\left(\mathcal{J}_{1}\right)$ denotes the $\sigma$-field generated by $\mathcal{J}_{1}$ and ${\cal B}\left(\mathbb{R}^{2p}\right)$ is the Borel $\sigma$-field on $\mathbb{R}^{2p}$. $\sigma\left(\mathcal{J}_{1}\right)={\cal B}\left(\mathbb{R}^{2p}\right)$ is a well known result. We show that all conditions of Lemma C1 in AS are satisfied. Then condition (\ref{eq:giv inequality_2}) holds from Lemma C1 in AS. 
\par
First, $\mathcal{J}_{1} $ is a semiring of subsets of $\mathbb{R}^{2p}$. Since $m\left(W_{i},W_{j},\beta\right)$ and $I\left[\left(X_{i}^{(1)},X_{j}^{(1)}\right)\in J_{1} \right]$ are bounded functions, $\mu(\cdot)$ satisfies the boundedness condition of Lemma C1 in AS. The other conditions of Lemma C1 in AS also hold by the same argument as the proof of Lemma 3 in AS. Thus, by applying Lemma C1 in AS, the fact that $\mu(J_{1})\geq0$ for all $J_{1} \in \mathcal{J}_{1}$ implies that $E_{P}\left[m\left(W_{i},W_{j},\beta\right)\cdot I\left[\left(X_{i}^{(1)},X_{j}^{(1)}\right)\in J_1 \mid (X_{i}^{(2)},X_{j}^{(2)})=(b,\tilde{b})\right]\right]\geq0$ for all $J_1 \in \sigma(\mathcal{J}_{1})$, which is equivalent to ${\cal B}\left(\mathbb{R}^{2p}\right)$. This implies that the result of Lemma \ref{lem:lemma_2} holds.
\end{proof}

\bigskip

Let $Q$ be a measure on $\MG$ such that, for $g_{(a,b),(\tilde{a},\tilde{b}),r} \in \MG$,
\begin{align*}
    Q(g_{(a,b),(\tilde{a},\tilde{b}),r}) = (r^{2}+100)^{-1} \cdot ((2r)^{p}\cdot |\MX_{2}|)^{-2}.
\end{align*}
Then $T_{n}(\beta)$ is expressed as 
\begin{align}
    T_{n}(\beta) = \int_{g \in \MG} S(n^{\frac{1}{2}}\bar{m}_{n}(\beta,g),\bar{\sigma}_{n}(\beta,g)) dQ(g),
    \label{eq:AR test statistics}
\end{align}
where $S(m,\sigma)\equiv [m/\sigma]_{-}^{2}$.

The following lemma shows that $Q$ satisfies Assumption Q in AS.

\begin{lemma}\label{lem:Assumption Q}
Suppose that Assumption \ref{asm:iid} holds. Define the pseud-metric $\rho_X$ on $\MG$ by
\begin{align*}
    \rho_{X}(g,g^{\ast}) \equiv (E_{P_X}[|g(X_i,X_j)-g^{\ast}(X_i,X_j)|])^{1/2} \mbox{ for }(g,g^{*}) \in \MG\times\MG,
\end{align*}
where $E_{P_X}[\cdot]$ denotes the expectation under $X \sim P_X$. Let $\MB_{\rho_X}(g,\delta)$ denote an open $\rho_{X}$-ball in $\MG$ centered at $g$ with radius $\delta$. 
Then the support of $Q$ under $\rho_{X}$ is $\MG$. That is, for all $\delta > 0$, $Q(\MB_{\rho_{X}}(g,\delta))>0$ for all $g \in \MG$.
\end{lemma}

\begin{proof}
We prove the statement in the case of $p=k$. Then the result for the other case $p<k$ immediately follows. Let $AR \equiv \{(a,\tilde{a},r): (a,\tilde{a}) \in \{1,\ldots,2r \}^{2k} \mbox{ and } r=1,2,\ldots\}$. Then there is a one-to-one mapping $\Pi:\MG \rightarrow AR$. The weight function $Q$ can be expressed as $Q=\Pi^{-1}Q_{AR}$, where $Q_{AR}$ is a measure on $2^{AR}$ such that $Q((a,\tilde{a},r))= (r^{2}+100)^{-1} \cdot (2r)^{-2k}$ for any $(a,\tilde{a},r) \in AR$. Then the statement follows by applying Lemma 4 in AS to $Q=\Pi^{-1} Q_{AR}$.
\end{proof}

\bigskip

\begin{proof}[Proof of Theorem \ref{thm:asymptotic property}]
To show that Theorem \ref{thm:asymptotic property} holds, it suffices to show that all required conditions in Theorems 2(b) and 3 in AS are satisfied.
We first check that the functions $S$ and $Q$ in (\ref{eq:AR test statistics}) satisfy the required conditions.
Since the function $S$ corresponds to the modified method moments function in AS (Equation (3.8) in AS), Lemma 1 in AS guarantees that the function $S$ in (\ref{eq:AR test statistics}) satisfies Assumptions S1--4 in AS.
Lemma \ref{lem:Assumption Q} shows that the weight function $Q$ in (\ref{eq:AR test statistics}) satisfies Assumption Q in AS.
Lemma \ref{lem:lemma_1} implies that Assumption EP in AS holds. Then Lemma B3 in AS guarantees that Assumption GMS2(a) in AS holds.
Assumption \ref{asm:tuning parameters} implies that the tuning parameters in the GMS function satisfy Assumptions GMS1, GMS2(b), and GMS2(c) in AS. 
Therefore, Theorem \ref{thm:asymptotic property}(a) follows from Theorem 2(b) in AS.

For result (b), Lemma \ref{lem:lemma_2} in this appendix has the same role as Assumption CI in AS.
The assumption $P((X_i,X_j) \in \MX_{P}^{2}(\tilde{\beta}))>0$ in Theorem \ref{thm:asymptotic property}(b) corresponds to Assumption FA(a) in AS.
Note that Lemma \ref{lem:lemma_1} applies for $(\beta_{a_n},P_{a_n})=(\tilde{\beta},P)$ by the same argument as in the proof of Lemma \ref{lem:lemma_1}.
Therefore, by the same arguments as in Section 14.2 in AS, where we replace Lemma A1 (in AS) for $(\beta_{a_n},P_{a_n})=(\tilde{\beta},P)$ by Lemma \ref{lem:lemma_1} for $(\beta_{a_n},P_{a_n})=(\tilde{\beta},P)$, Theorem \ref{thm:asymptotic property}(b) holds.
\end{proof}

\bigskip
\section{Inference for the Transformation Function \label{app:joint_inference}}

This section provides a joint inference procedure for $\beta_{0}$ and $\{T\left(y\right)\}_{y \in \By}$, where $\By$ is a finite set of $y$. Let $q$ be the dimension of $\By$ and denote $\By = (y_{1},\ldots,y_{q})^{\prime}$.
We first note that with fixed $y \in \Real$ and $\beta$, $T_{I,\beta}(y)$ is equivalent to the set of all values $t \in \Real$ that satisfy the following conditional moment inequality condition:
\begin{align}
    E_{P}\left[m^{\dagger}(W_{i},W_{j},\beta,y,t)\mid x_{i},x_{j}\right] \geq 0 
    \mbox{ for all $(x_{i},x_{j}) \in \MX^{2}$}, 
    \label{eq:conditional moment inequality_T}
\end{align}
where $m^{\dagger}(W_{i},W_{j},\beta,y,t)$ is a moment function defined as
\begin{align*}
m^{\dagger}\left(W_{i},W_{j},\beta,y,t\right) & \equiv  \left(I\left[Y_{1i}\geq y\right]-I\left[Y_{0j}\geq\tilde{y}\right]\right)\cdot I[X_{i}^{\prime}\beta-X_{j}^{\prime}\beta\geq t]\\
 & +\left(I\left[Y_{1j}\geq y\right]-I\left[Y_{0i}\geq\tilde{y}\right]\right)\cdot I[X_{j}^{\prime}\beta-X_{i}^{\prime}\beta\geq t].
\end{align*}
This is obvious from the definition of $T_{I,\beta}(y)$. 
The joint inference procedure for $\beta_0$ and $\{T\left(y\right)\}_{y \in \By}$
is then constructed by applying the conditional moment inequality inference approach presented in Section \ref{sec:inference} to the conditional moment inequalities (\ref{eq:conditional moment inequality}) and (\ref{eq:conditional moment inequality_T}).
\par

Fix $y \in \By$. Let $m^{\dagger}\left(W_{i},W_{j},\beta,y,t,g\right) \equiv m^{\dagger}\left(W_{i},W_{j},\beta,y,t\right)\cdot g(x_{i},x_{j})$
for any $g\in\mathcal{G}$, where $\MG$ is the set of instrumental functions defined in Section \ref{sec:test statistics}. Then $T_{I,\beta}\left(y\right)$ is equivalent to
\begin{equation*}
\left\{ t\in\mathbb{R}:\ E_{P}\left[m^{\dagger}\left(W_{i},W_{j},\beta,y,t,g\right)\right]\geq0\ \mbox{for all}\ g\in\mathcal{G}\right\} .
\end{equation*}

Define the sample moment function and sample variance function, respectively, by 
\begin{align*}
\bar{m}_{n}^{\dagger}\left(\beta,y,t,g\right) \equiv \frac{1}{n\left(n-1\right)}\sum_{i\neq j}m^{\dagger}\left(W_{i},W_{j},\beta,y,t,g\right)
\end{align*}
and 
\begin{align*}
\hat{\sigma}_{n}^{\dagger2}\left(\beta,y,t,g\right) \equiv & \left\{ \frac{1}{n\left(n-1\right)\left(n-2\right)}\sum_{i\neq j\neq k}m^{\dagger}\left(W_{i},W_{j},\beta,y,t,g\right)m^{\dagger}\left(W_{i},W_{k},\beta,y,t,g\right)\right.\\
 &  \left.-\left(\frac{1}{n(n-1)}\sum_{i\neq j}m^{\dagger}\left(W_{i},W_{j},\beta,y,t,g\right)\right)^{2}\right\} .
\end{align*}
Since the function $\bar{m}_{n}^{\dagger}\left(\beta,y,t,g\right)$ is a U-statistic of order two, the estimator of its asymptotic variance, $\hat{\sigma}_{n}^{\dagger2}\left(\beta,y,t,g\right)$, is constructed with a similar form to $\hat{\sigma}_{n}^{2}\left(\beta,g\right)$ in Section \ref{sec:test statistics}. In practice, we use the modified sample variance function: 
\begin{align*}
\bar{\sigma}_{n}^{\dagger2}\left(\beta,y,t,g\right)  \equiv  \hat{\sigma}_{n}^{\dagger2}\left(\beta,y,t,g\right)+\epsilon\hat{\sigma}_{n}^{\dagger2}(y),
\end{align*}
where $\hat{\sigma}_{n}^{\dagger2}(y)\equiv\hat{\sigma}_{n}^{\dagger2}\left(\beta,y,t,1\right)$
and $\epsilon$ is the regularization parameter (e.g., $\epsilon=0.0001$).
\par
Then, for any $\beta \in B$ and $\Bt=(t_1,\ldots,t_{q})^{\prime} \in \Real^{q}$, an approximate test statistic for a null hypothesis:
\begin{align*}
    H_{0}:&\ E_{P}\left[m(W_{i},W_{j},\beta)\mid x_{i},x_{j}\right] \geq 0 \mbox{ and }
    E_{P}\left[m^{\dagger}(W_{i},W_{j},\beta,y_{\ell},t_{\ell}) \mid x_{i}, x_{j}\right] \geq 0 \\
    & \mbox{for all }\ell=1,\ldots,q \mbox{ and } (x_{i},x_{j}) \in \MX^{2}
\end{align*}
is constructed as
\begin{align*}
T_{\By,n,R}^{\dagger}(\beta,\Bt) \equiv &\ \sum_{r=1}^{R}\left(r^{2}+100\right)^{-1}
\sum_{\left(a,\tilde{a}\right)\in \left\{ 1,\cdots,2r\right\}^{2p}}
\sum_{(b,\tilde{b})\in \MX_{2}^{2}}
\frac{1}{(2r)^{2p} \cdot |\MX_{2}|^{2}}
\left(\left[\frac{n^{\frac{1}{2}}\bar{m}_{n}\left(\beta,g_{(a,b),(\tilde{a},\tilde{b}),r}\right)}{\bar{\sigma}_{n}\left(\beta,g_{(a,b),(\tilde{a},\tilde{b}),r}\right)}\right]_{-}^{2}\right.\\
 &  \left.+ \sum_{j=1}^{q}\left[\frac{n^{\frac{1}{2}}\bar{m}_{n}^{\dagger}\left(\beta,y_{j},t_{j},g_{(a,b),(\tilde{a},\tilde{b}),r}\right)}{\bar{\sigma}_{n}^{\dagger}\left(\beta,y_{j},t_{j},g_{(a,b),(\tilde{a},\tilde{b}),r}\right)}\right]_{-}^{2}\right),
\end{align*}
where $R$ is a truncation integer chosen by the researcher. Note that this test statistic comprises a normalized sample moment function for $\beta$, $\bar{m}_{n}\left(\beta,g_{(a,b),(\tilde{a},\tilde{b}),r}\right)/\bar{\sigma}_{n}\left(\beta,g_{(a,b),(\tilde{a},\tilde{b}),r}\right)$, and $q$ normalized sample moment functions for $T(y_j)$, $\bar{m}_{n}^{\dagger}\left(\beta,y_{j},t,g_{(a,b),(\tilde{a},\tilde{b}),r}\right)/\bar{\sigma}_{n}^{\dagger}\left(\beta,y_{j},t,g_{(a,b),(\tilde{a},\tilde{b}),r}\right)$, with $j=1,\ldots,q$.
\par
We can compute the critical value for $T_{\By,n,R}^{\dagger}(\beta,\Bt)$ as a simulated quantile of
\begin{align*}
T_{\By,n,R}^{\dagger,Asy}(\beta,\Bt)  \equiv &\ \sum_{r=1}^{R}\left(r^{2}+100\right)^{-1}
\sum_{\left(a,\tilde{a}\right)\in \left\{ 1,\cdots,2r\right\}^{2p}}
\sum_{(b,\tilde{b})\in \MX_{2}^{2}}
\frac{1}{(2r)^{2p} \cdot |\MX_{2}|^{2}} \\
 & \times \left(\left[\frac{v_{n}\left(\beta,g_{(a,b),(\tilde{a},\tilde{b}),r}\right)+\varphi_{n}\left(\beta,g_{(a,b),(\tilde{a},\tilde{b}),r}\right)}{\bar{\sigma}_{n}\left(\beta,g_{(a,b),(\tilde{a},\tilde{b}),r}\right)}\right]_{-}^{2}\right.\\
   & \left.+\sum_{j=1}^{q}\left[\frac{v_{n}^{\dagger}\left(\beta,y_{j},t_{j},g_{(a,b),(\tilde{a},\tilde{b}),r}\right)+\varphi_{n}^{\dagger}\left(\beta,y_{j},t_{j},g_{(a,b),(\tilde{a},\tilde{b}),r}\right)}{\bar{\sigma}_{n}^{\dagger}\left(\beta,y_{j},t_{j},g_{(a,b),(\tilde{a},\tilde{b}),r}\right)}\right]_{-}^{2}\right),
\end{align*}
where $\left(v_{n}^{\dagger}\left(\beta,y,t,g\right)\right)_{g\in{\cal G}}$ is a zero mean Gaussian process with a covariance kernel evaluated by
\begin{eqnarray*}
\hat{h}_{2}^{\dagger}\left(\beta,y,t,g,g^{\ast}\right) & \equiv & \left\{ \frac{1}{n\left(n-1\right)\left(n-2\right)}\sum_{i\neq j\neq k}m^{\dagger}\left(W_{i},W_{j},\beta,y,t,g\right)m^{\dagger}\left(W_{i},W_{k},\beta,y,t,g^{\ast}\right)\right.\\
 &  & \left.-\left[\frac{1}{n\left(n-1\right)}\sum_{i\neq j}m^{\dagger}\left(W_{i},W_{j},\beta,y,t,g\right)\right]\cdot\left[\frac{1}{n\left(n-1\right)}\sum_{i\neq j}m^{\dagger}\left(W_{i},W_{j},\beta,y,t,g^{\ast}\right)\right]\right\},
\end{eqnarray*}
and $\varphi_{n}^{\dagger}\left(\beta,y,t,g_{(a,b),(\tilde{a},\tilde{b}),r}\right)$ is a GMS function given by
\begin{equation*}
\varphi_{n}^{\dagger}\left(\beta,y,t,g_{(a,b),(\tilde{a},\tilde{b}),r}\right)\equiv\hat{\sigma}_{n}^{\dagger2}(y)B_{n}I\left[\kappa_{n}^{-1}n^{\frac{1}{2}}\bar{m}_{n}^{\dagger}\left(\beta,y,t,g_{(a,b),(\tilde{a},\tilde{b}),r}\right)/\bar{\sigma}_{n}^{\dagger}\left(\beta,y,t,g_{(a,b),(\tilde{a},\tilde{b}),r}\right)>1\right].
\end{equation*}
Here $B_{n}$ and $\kappa_{n}$ are two tuning parameters that should satisfy Assumption \ref{asm:tuning parameters}. (e.g., $\kappa_{n}=\left(\left(1-\hat{p}_{1-D}^{1/3}\right)^{2/5}\times0.6\ln(n)\right)^{\frac{1}{2}}$
and $B_{n}=\left(0.8\ln\left(n\right)/\ln\ln\left(n\right)\right)^{\frac{1}{2}}$).
For a significance level of $\alpha<1/2$, let $\hat{c}_{\By,\eta,1-\alpha}\left(\beta,\Bt\right)$ be the $1-\alpha+\eta$ sample quantile of $T_{\By,n,R}^{Asy}(\beta,\Bt)$ given an arbitrarily small value $\eta$ (e.g., $\eta=10^{-6}$). Then
the $\left(1-\alpha\right)$-level confidence set for $\left(\beta_{0},T\left(y_{1}\right),\ldots,T\left(y_{q}\right)\right)$
is computed as
\begin{equation*}
\left\{ \left(\beta,\Bt\right)\in B\times\Real^{q}:T_{\By,n,R}(\beta,\Bt)\leq\hat{c}_{\By,\eta,1-\alpha}\left(\beta,\Bt\right)\right\} .
\end{equation*}
The size and power properties stated in Theorem \ref{thm:asymptotic property} should apply to this confidence set. Monte Carlo simulation in Section D in the supplementary material examines the finite sample performance of this joint inference procedure.


\bigskip

\bibliographystyle{ecta}
\bibliography{ref_censoring}


\end{spacing}

\end{document}